\documentclass[12pt, draftclsnofoot, onecolumn]{IEEEtran}
\usepackage[colorlinks,
linkcolor=red,
anchorcolor=green,
citecolor=blue
]{hyperref} 
\usepackage{amsfonts,amssymb,amsmath,array}
\usepackage{latexsym}
\usepackage{CJK}
\usepackage{cite}
\usepackage{bm}
\usepackage{soul}

\usepackage{graphicx}
\usepackage{epstopdf}
\usepackage{epsfig}
\usepackage{subfigure} 

\usepackage{verbatim}  

\usepackage[usenames,dvipsnames]{color}
\definecolor{mycyan}{cmyk}{.3,0,0,0}
\usepackage{colortbl}
\usepackage{enumerate}
\usepackage{cases}

\usepackage{mathrsfs}   
\usepackage{algorithm} 
\usepackage{algorithmic} 
\usepackage{booktabs}
\usepackage{textcomp}
\usepackage{multirow}
\usepackage{lettrine}
\usepackage[scaled=0.92]{helvet}

\usepackage{algorithm}
\usepackage{algorithmic}

\usepackage{amsthm} 

\newtheorem{lem}{Lemma}
\newtheorem{prop}{Proposition}
\newtheorem{defn}{Definition}
\newtheorem{rem}{Remark}

\begin{document}
\title{Non-parametric Message Important Measure: Storage Code Design and Transmission Planning for Big Data}

\author{Shanyun~Liu, Rui~She, Pingyi~Fan,~\IEEEmembership{Senior Member, ~IEEE}, Khaled~B.~Letaief,~\IEEEmembership{Fellow, ~IEEE}\\

\thanks{
Shanyun~Liu, Rui~She and Pingyi Fan are with Tsinghua National Laboratory for Information Science and Technology(TNList) and the Department of Electronic Engineering,  Tsinghua University,  Beijing,  P.R. China, 100084. e-mail: \{liushany16,~sher15@mails.tsinghua.edu.cn, fpy@tsinghua.edu.cn.\} }
\thanks{Khaled~B.~Letaief is with the Department of Electrical and Computer Engineering, HKUST, Hong Kong. email:
\{eekhaled@ust.hk\} and Hamad bin Khalifa University, Qatar. email: \{kletaief@hbku.edu.qa.\}
}
\thanks{This work was supported in part by the National Natural Science Foundation of China (NSFC) under Grant 61771283, and in part by the China Major State Basic Research Development Program (973 Program) under Grant 2012CB316100(2).}
}

\maketitle

\graphicspath{{Figures/}}

\begin{abstract}

Storage and transmission in big data are discussed in this paper, where message importance is taken into account. Similar to Shannon Entropy and Renyi Entropy, we define non-parametric message important measure (NMIM) as a measure for the message importance in the scenario of big data, which can characterize the uncertainty of random events. It is proved that the proposed NMIM can sufficiently describe two key characters of big data: rare events finding and large diversities of events. Based on NMIM, we first propose an effective compressed encoding mode for data storage, and then discuss the channel transmission over some typical channel models. Numerical simulation results show that using our proposed strategy occupies less storage space without losing too much message importance, and there are growth region and saturation region for the maximum transmission, which contributes to designing of better practical communication system.

\end{abstract}


\IEEEpeerreviewmaketitle

\section{Introduction}\label{Sec:Introduction}

Recent studies show that the problem of data storage and message transmission is becoming much more important and intractable in the era of big data for the contradiction between the limited traditional hardware storage equipment and the sharply increasing data deluge \cite{chen2014big}. On the one hand, it is usually neither practical nor worth to store all the message as before when facing the huge amount of data, such as the massive number of connected devices in the Internet of Things \cite{Gharbieh2017Spatiotemporal}. On the other hand, the exponentially increasing data traffic will present huge challenges to transmission design in the future \cite{bi2015wireless}. One of the effective solutions for storing and transmitting the massive amount of data is to directly compress the data sample first and then recover them from compressed samples with minimized information loss \cite{pourkamali2017preconditioned}. In fact, there is a high level of redundancy in datasets, so redundancy reduction and data compression have little effect on the potential values of the data \cite{chen2014big}. The lossy data compression has been discussed in many scenarios, e.g., \cite{aguerri2016lossy,cui2012distributed,jalali2012block,sechelea2016rate}. In addition, the lossy transmission of a bivariate Gaussian sources over a Gaussian broadcast channel was studied in \cite{koken2017joint}. It is noted that the same amount of information loss of different events is treated equivalent in these references. However, it may have various costs in different events. While the errors in useless message may have little impact, the errors in important message can bring disastrous consequences. 

Actually, in many cases, sparse events are exceedingly crucial for the fact that people generally concentrate more on the important part of message rather than the whole message itself. For instance, only a few illicit IDs need to be checked in order to prevent financial frauds in synthetic ID detection \cite{phua2010comprehensive}. Also, only a few individuals need to be closely checked in anti-terrorist system \cite{zieba2015counterterrorism}. From ancient times, human beings preferred to record when and where nature disaster events happened (i.e. earthquakes, mountain torrents, hurricanes) rather than to record much more common things of daily life. In this case, the rare events usually contain the most information. Therefore, the most important information will be substantially retained if the minority subset of these rare events is recorded, which coincides with the cognitive mechanism of human beings. Generally speaking, people would like to take the events with high probability for granted and would not like to waste storage or transmission resource for them. In fact, the information itself may not have difference from the viewpoint of transmission and storage, but it will be given different importance when received by individuals and/or organizations. Thus, in this paper, we design the strategy of storage and transmission, in which the data is dynamically compressed based on the importance of events, in the case where the rare events contain the most information, which is the minority subset stated above.

Previous studies on the minority subset mainly focused on minority subset detecting under rate-distortion theory \cite{crammer2004needle,ando2006information}, graph-based rare category detection \cite{he2008graph}, hypothesis testing \cite{katz2017distributed}, or detecting with the aid of time-evolving of graphs \cite{zhou2015rare}. Information theory has been considered as one of the most important theories in the fields of data storage and communications \cite{Elements,verdu1998fifty}. Information theory gives the optimal coding and a tight bound of lossless data compression under the typical set \cite{Elements}. However, these results can not be applied to lossy data compression because it is derived in lossless coding. The other existing information measures, such as Renyi divergence, f-divergence or Fisher information, are usually used in communication theory, data analysis, hypothesis testing or estimation of statistical parameter \cite{van2014renyi,akaike1998information}. However, all these existing information measures characterize all the information as a whole, focusing on typical sets, the events with higher probabilities, rather than the minority subset of events. Therefore, a specific information measure is needed to characterize the minority subsets.

Message Important Measure (MIM) is a fairly useful tool to describe the message importance, which is commonly used in minority subset detection as an information measure \cite{fan2016message}. The definition is given by \cite{fan2016message}
\begin{equation}
  L\left( {{\textbf{p}},\varpi } \right) = {L_\varpi }({\textbf{p}}) = \log \sum\limits_{i = 1}^n {{p_i}\exp \left\{ {\varpi \left( {1 - {p_i}} \right)} \right\}},
\end{equation}
where ${\textbf{{p}}}=(p_1,p_2,...,p_n)$ is a given probability distribution and parameter $\varpi$ is the importance coefficient.  By focusing on those small-probability events, MIM is very effective in minority subset detection \cite{fan2016message}. The parameter selection of MIM in big data is discussed in \cite{she2017focusing}. In fact, these two literatures both examined parametric MIM. The flexible parameter in MIM helps to emphasize a certain probability element in a distribution, which is very useful in anomaly detection. However, this parametric MIM may not be applicable in some scenarios for the fact that the optimization of the parameter needs some prior knowledge and also takes excessive computing resources, which makes the MIM in \cite{fan2016message} not universal. To overcome this problem, we expand MIM to non-parametric MIM (NMIM), which is more general. Some basic properties and applications was preliminarily discussed in \cite{liu2017non}. The more general properties, including NMIM loss distortion and more applications, including the compressed storage code design and transmission planning will be discussed in this paper.

Similar to Shannon entropy, NMIM is only related to the probability distribution. Therefore, it can be easily utilized without the complex process of parameter determination. For convenience of calculations, NMIM is defined as a logarithm of the mean value of the events' importance. Actually, there are many interesting properties for this new information measure. One important result is that the gap between NMIM of total events and that of minimum probability event will be less than a constant if either sparsity or diversity is satisfied.

Actually, layered transmission schemes and layered coding have been widely discussed in the communication systems designing, e.g. \cite{katz2017distributed,Bhattad2008distortion,wang2016two,ethakaset2011joint}. In these references, different layers perform disparate things by using different information. This paper proposes two-layer coding mechanism, where the first layer performs the amount of information by Shannon entropy, and the second layer uses NMIM as importance marking. 

We then discuss, by means of NMIM, the problem of compressed storage and provide strategy that enormously reduces the precious storage space. This is possible because of allocating resources rationally by providing more code length for important events and little code length for nonsignificant events. 

In order to design better transmission strategy, we analyze the change of NMIM during transmission. Then, an achievable NMIM-loss-distortion region is also derived, which is ternary, comprised of the NMIM loss, the distortion and the events' probability distribution. Based on this NMIM-loss-distortion region, one solution of the maximum transmission is proposed which turns to maximizing the entropy in the receiving terminal with limited importance loss.

The main contribution of this paper can be summarized as follows. We give the definition and properties of NMIM, which is an effective measure of importance information. Based on NMIM, a compressed storage strategy in big data is proposed. Furthermore, we analyze the change of NMIM during the transmission and propose the NMIM-loss-distortion region. With the developed results, we solve the problem of the maximum transmission.

The rest of this paper is organized as follows. Section II introduces the definition of NMIM. Section III gives a detailed discussion of its properties. The speciality of minimum probability is also investigated. In Section IV, the compressed coding problem is formulated as an optimization problem and the explicit solution is presented. The performance of transmission is discussed by means of NMIM to illustrate the practicality of NMIM in Section V. We also study the maximum transmission problem with limited channel capacity and importance loss in this section. Section VI presents the simulation results to certificate our developed theories in this paper and discusses the numerical results. Finally, we present the conclusion in Section VII.

\section{Non-parametric Message Important Measure}
In this section, a new non-parametric information measure will be introduced, which is called NMIM. Different from previous studies, in which Shannon entropy uses logarithm operator and Renyi divergence uses polynomial operator \cite{van2014renyi}, NMIM adopts the exponential form because exponential operator can help it magnify the minority subsets.

 \begin{defn}\label{Non-parametric MIM}
\par For a given probability distribution ${\textbf{{p}}}=(p_1,p_2,...,p_n)$ of finite alphabet (each element $0<p_i \le 1$), the non-parametric message importance message is defined as
\begin{equation}
\mathscr{L}\left( {\textbf{{p}}} \right) = \log \sum\limits_{i = 1}^n {{p_i}{e^{{{1 - {p_i}} \over {{p_i}}}}}} .
\end{equation}
\end{defn}

This equation can also be written as
\begin{equation}
\mathscr{L}\left( {\textbf{\emph{p}}} \right) = \log \sum\limits_{i = 1}^n {{p_i}{e^{{1 \over {{p_i}}}}}}  - \log e.
\end{equation}

In fact, ${e^{{{1 - {p_i}} \over {{p_i}}}}}$ is defined as the importance measure of the event with probability of $p_i$. In some case where we need to discuss the importance of one specific event, the following definition is given for convenience.

\begin{defn}\label{MIM of Event}
The importance of the event whose probability is $p_i$ is given by
\begin{equation}\label{equdefn:MIM of Event}
M(p_i)=p_i e^{{1-p_i} \over {p_i}}.
\end{equation}
\end{defn}
As a result, NMIM can also be written as
\begin{equation}\label{equ:MIM log defn}
\mathscr{L}\left( {\textbf{\emph{p}}} \right) = \log \sum\limits_{i = 1}^n M(p_i).
\end{equation}
The total importance is the sum of importance of all events, and the logarithm operator in (\ref{equ:MIM log defn}) is only to reduce the magnitude of the numerical results.

Let $L(p_i)$ be the logarithm value of the importance of a specific event with probability $p_i$, and it can be written as
\begin{equation}
L\left( p_i \right) =\log M(p_i)= \log  ({{p_i}{e^{{{1 - {p_i}} \over {{p_i}}}}}} ).
\end{equation}

\section{The Properties of NMIM}
In this section, the properties of NMIM is discussed in details. Natural logarithm is adopted in this part.

\subsection{ Non-negativity}
It is obvious that $0 < p_i \leq 1$, so $e^{\frac{1-p_i}{p_i}}\geq 1$ for the fact that $\frac{1-p_i}{p_i}\geq0$ and we have
\begin{equation}
\mathscr{L}\left( {\textbf{\emph{p}}} \right) = \log \sum\limits_{i = 1}^n {{p_i}{e^{{{1 - {p_i}} \over {{p_i}}}}}}  \ge \log \sum\limits_{i = 1}^n {{p_i}}  = 0,
\end{equation}
which means NMIM is non-negativity.

\subsection{Uniform Distribution}
For the uniform distribution $u=(1/n,1/n,...,1/n)$, we can obtain
\begin{equation}
\mathscr{L}\left( {\textbf{\emph{u}}} \right) = \log \sum\limits_{i = 1}^n {{1 \over n}{e^{{{1 - {1 \over n}} \over {{1 \over n}}}}}}  = \log \sum\limits_{i = 1}^n {{1 \over n}{e^{n - 1}}}  = n - 1 .
\end{equation}


\subsection{ Lower bound}
For any probability distribution ${\textbf{{p}}}=(p_1,p_2,...,p_n)$ without zero elements, it is noted that
\begin{equation}\label{equ:lowerbound}
\begin{split}
\mathscr{L}\left( {\textbf{\emph{p}}} \right)& = \log \sum\limits_{i = 1}^n {{p_i}{e^{{{1 - {p_i}} \over {{p_i}}}}}}  \ge \log {e^{\sum\limits_{i = 1}^n {{p_i}} {{1 - {p_i}} \over {{p_i}}}}}= \log {e^{\sum\limits_{i = 1}^n {(1 - {p_i})} }} = n - 1,
 \end{split}
\end{equation}
where the inequality is according to Jensen's inequality \cite{Elements}. The equality holds if and only if all the $p_i$ are equal, i. e. $p_i={1}/{n}$, which means the uniform distribution achieves the lower bound.



\subsection{Monotonicity}\label{sec:monotonicity}
For two probability $p_1$ and $p_2$, if $p_1<p_2$, then we will have $L(p_1)>L(p_2)$.
\begin{proof}
Define $f(x) = x{e^{{1 -x\over x }}}$ and the derivation is given by $f'(x) = \left( {1 - {1 \over x}} \right){e^{{1-x \over x}}}$.
When $0<x<1$, it is easy to check that $f'(x)<0$, so that $f(x)$ is a monotonic decreasing function. Since the outer logarithmic function does not change the function monotonicity, we have $L(p_1)>L(p_2)$ if $p_1<p_2$.
\end{proof}

\subsection{Geometric Center and Barycenter}
For any probability distribution ${\textbf{{p}}}=(p_1,p_2,...,p_n)$ without zero elements, it is noted that
\begin{equation}\label{equ:G and B Center}
\sum\limits_{i = 1}^n {{p_i}L\left( {{p_i}} \right)}  \le \sum\limits_{i = 1}^n {{1 \over n}L\left( {{p_i}} \right)} 
\end{equation}
\begin{proof}
\begin{flalign}\label{equ:G and B Center1} 
\sum\limits_{i = 1}^n {{1 \over n}L\left( {{p_i}} \right) }-& \sum\limits_{i = 1}^n {{p_i}L\left( {{p_i}} \right)}  = \sum\limits_{{p_i} \le 1/n} {\left( {{1 \mathord{\left/
 {\vphantom {1 n}} \right.
 \kern-\nulldelimiterspace} n} - {p_i}} \right)L({p_i})}  - \sum\limits_{{p_i} > 1/n} {\left( {{p_i} - {1 \mathord{\left/
 {\vphantom {1 n}} \right.
 \kern-\nulldelimiterspace} n}} \right)L({p_i})}  \\
& \ge \sum\limits_{{p_i} \le 1/n} {\left( {{1 \mathord{\left/
 {\vphantom {1 n}} \right.
 \kern-\nulldelimiterspace} n} - {p_i}} \right)L({1 \mathord{\left/
 {\vphantom {1 n}} \right.
 \kern-\nulldelimiterspace} n})}  - \sum\limits_{{p_i} > 1/n} {\left( {{p_i} - {1 \mathord{\left/
 {\vphantom {1 n}} \right.
 \kern-\nulldelimiterspace} n}} \right)L({1 \mathord{\left/
 {\vphantom {1 n}} \right.
 \kern-\nulldelimiterspace} n})}   \tag{\theequation a}\label{G and B Center1 a}\\
&= \left( {\sum\limits_{{p_i} \le 1/n} {{1 \mathord{\left/
 {\vphantom {1 n}} \right.
 \kern-\nulldelimiterspace} n}}  - \sum\limits_{{p_i} \le 1/n} {{p_i}}  - \sum\limits_{{p_i} > 1/n} {{p_i}}  + \sum\limits_{{p_i} > 1/n} {{1 \mathord{\left/
 {\vphantom {1 n}} \right.
 \kern-\nulldelimiterspace} n}} } \right)L({1 \mathord{\left/
 {\vphantom {1 n}} \right.
 \kern-\nulldelimiterspace} n}) = 0  \tag{\theequation b}\label{G and B Center1 b}
\end{flalign}
where  \eqref{G and B Center1 a} follows from the monotonicity of $L(\cdot)$. The equality holds if and only if $p_i={1}/{n},i=1,2,...,n$.
\end{proof}

\subsection{Event Decomposition and Merging}
Let ${p_i}^{(1)} + {p_i}^{(2)} = {p_i}$, which are the probabilities of the first and second sub-events of the i-th event,  we have
\begin{equation}\label{equ:event decom and merge}
{p_i}^{(1)}{e^{{{1 - {p_i}^{(1)}} \over {{p_i}^{(1)}}}}} + {p_i}^{(2)}{e^{{{1 - {p_i}^{(2)}} \over {{p_i}^{(2)}}}}} \ge {p_i}{e^{{{1 - {p_i}} \over {{p_i}}}}}.
\end{equation}
This shows that NMIM will be increased when one event is divided into two sub-events. On the other hand, if two events are merged into one event, NMIM will be decreased. This is to say, more observation knowledge leads to the increasing of event predication accuracy or decreasing of event unexpected degree.
\begin{proof}
\begin{equation}
\begin{split}
  &{p_i}^{(1)}{e^{{{1 - {p_i}^{(1)}} \over {{p_i}^{(1)}}}}} + {p_i}^{(2)}{e^{{{1 - {p_i}^{(2)}} \over {{p_i}^{(2)}}}}} \ge {p_i}^{(1)}{e^{{{1 - {p_i}} \over {{p_i}}}}} + {p_i}^{(2)}{e^{{{1 - {p_i}} \over {{p_i}}}}} = \left( {{p_i}^{(1)} + {p_i}^{(2)}} \right){e^{{{1 - {p_i}} \over {{p_i}}}}} = {p_i}{e^{{{1 - {p_i}} \over {{p_i}}}}} .\cr
\end{split}
\end{equation}
\end{proof}


\subsection{Indepent Probability Distributions}
For two given probability distribution ${\textbf{{p}}}=(p_1,p_2,...,p_n)$ and ${\textbf{{q}}}=(q_1,q_2,...,q_n)$ of finite alphabet without zero elements (each element $0<p_i\le1$, $0<q_i\le1$), we observe that
\begin{equation}
\mathscr{L}({\textbf{\emph{p}}}) + \mathscr{L}({\textbf{\emph{q}}}) \le\mathscr{L}({\textbf{\emph{pq}}})
\end{equation}
\begin{proof}
\begin{flalign}\label{equ:Indepent Probability Distributions} 
  & \mathscr{L}({\textbf{\emph{p}}}) + \mathscr{L}({\textbf{\emph{q}}}) = \log \sum\limits_{i = 1}^n {{p_i}{e^{{1-p_i \over {{p_i}}}}}}  + \log \sum\limits_{j = 1}^n {{q_j}{e^{{1-q_j \over {{q_j}}}}}}     = \log \sum\limits_{i = 1}^n {{p_i}{e^{{1-p_i \over {{p_i}}}}}} \sum\limits_{j = 1}^n {{q_j}{e^{{1-q_j \over {{q_j}}}}}} \\ 
 & = \log \sum\limits_{i,j} {{p_i}{q_j}{e^{{1-p_i \over {{p_i}}} + {1-q_j \over {{q_i}}}}}}   = \log \sum\limits_{i,j} {{p_i}{q_j}{e^{{{{p_i} +  {q_j}-2 p_i q_j} \over {{p_i}{q_j}}}}}} \overset{(*)}\le \log \sum\limits_{i,j} {{p_i}{q_j}{e^{{{1 -p_i q_j} \over {{p_i}{q_j}}}}}}  = \mathscr{L}({\textbf{\emph{pq}}}). \tag{\theequation a}\label{equ:Indepent Probability Distributions}
\end{flalign}
The inequation (*) can be verified by $p_i+q_j-p_i q_j \le 1$ since $1 - {p_i} - {q_j} + {p_i}{q_j} = \left( {1 - {p_i}} \right)\left( {1 - {q_j}} \right) \ge 0$.
\end{proof}

\subsection{Taylor Series Expansion}
Obviously, the $n$th order derivative of $L(p_i)$ is:
\begin{equation}\label{equ:nth L_de}
L^{(n)}\left( {{p_i}} \right) = {(-1)}^{n+1}{(n-1)!} { 1\over {{p_i}}^{n}}  +{(-1)}^{n} {n!}{ 1\over {{p_i}}^{n+1}}.
\end{equation}
Taylor series expansion of $L(p_i)$ is given by
\begin{equation}\label{equ:taylor series expansion}
\begin{split}
L\left( {{p_i} + \Delta p} \right) = L\left( {{p_i}} \right) + \sum_{k=1}^n {{( - 1)^k}{{k - {p_i}} \over {kp_i^{k + 1}}}\Delta {p^k}}  + o(\Delta {p^n}).
\end{split}
\end{equation}
In particular, if $p_i \ll 1 <n$, then it can be written as
\begin{equation}
\begin{split}
L\left( {{p_i} + \Delta p} \right) = L\left( {{p_i}} \right) +\sum_{k=1}^n {{{(-1)}^{k}\over p_i^{k+1}} \Delta p^k}+ o(\Delta {p^n}).
\end{split}
\end{equation}


\subsection{Minimum probability}
\begin{lem}\label{lem:root}
If $0 < {p_1} < {p_2} < {1 \over n}$ $(n\ge2)$ and $n$ is large enough, such as $n \gg 1$, then we have
\begin{equation}
    {{M({p_1})} \over {M({p_2})}}   \ge n - 1.
\end{equation}
\end{lem}
\begin{proof}
Refer to the Appendix \ref{Appendices A}.
\end{proof}

\begin{rem}\label{rem:big n}
Lemma \ref{lem:root} implies that no matter how close or far the two probabilities $p_1$ and $p_2$ are, there always exists a $N$, when $n>N$, we have ${{M({p_1})} \over {M({p_2})}} \ge n-1$. 
\end{rem}

\begin{lem}\label{lem:root1}
For $0 < {p_1} < {p_2} < {1 \over n}$ $(n\ge2)$, if ${p_1} \ll {1 \over {1 + \ln (n - 1)}}$,
then ${{M({p_1})} \over {M({p_2})}} \ge n-1.$
\end{lem}
\begin{proof}
Refer to the Appendix \ref{Appendices B}.
\end{proof}

\begin{rem}\label{rem:min}
 Lemma \ref{lem:root1} means that there exists a $p_0$, and if $p_1<p_0$, ${{M({p_1})} \over {M({p_2})}} \ge n-1$ is always true for $n \ge 2$.
\end{rem}

\begin{prop}\label{thm:summax}
For a given non-zero probability distribution ${\textbf{{p}}}=(p_1,p_2,...,p_n)$, let ${p_{{\min}}}$ and ${p_{{s\min}}}$ be the minimum and the second minimum values in the distribution ${\textbf{{p}}}$ ($p_{min}<p_{s\min}$), respectively. If it satisfies either of the following conditions:

i) $ p_{min} \ll {1 \over {1 + \ln (n - 1)}}$;

ii) The events size $n$ is large enough, such as $n \gg1$;\\
then we have
 \begin{equation}\label{equ:summax}
0 \le \mathscr{L}({\textbf{{p}}}) - L(p_{\min}) \le \log 2.
\end{equation}

\end{prop}
\begin{proof}
It is noted that
\begin{flalign}\label{equ:mini pro} 
\mathscr{L}({\textbf{\emph{p}}})   &   = \log \left( {{p_{\min }}{e^{{1-p_{\min} \over {{p_{\min }}}}}} + \sum\limits_{{p_i} \ne {p_{\min }}} {{p_i}{e^{{1-p_i \over {{p_i}}}}}} } \right) \\
&\le \log \left( {{p_{\min }}{e^{{1 -p_{\min}\over {{p_{\min }}}}}} + \sum\limits_{{p_i} \ne {p_{\min }}} {{p_{s\min}}{e^{{1- p_{s\min}\over {{p_{s\min}}}}}}} } \right) \tag{\theequation a}\label{equ:mini proa}\\
&= \log \left( {{p_{\min }}{e^{{1-p_{\min} \over {{p_{\min }}}}}} + (n - 1){p_{s\min }}{e^{{1 -p_{s\min}\over {{p_{s\min }}}}}}} \right)   \tag{\theequation b}\label{equ:mini prob}\\
&\le \log \left( {2{p_{\min }}{e^{{1-p_{\min} \over {{p_{\min }}}}}}} \right)   = L(p_{\min}) + \log 2.  \tag{\theequation c}\label{equ:mini proc}
\end{flalign}
(\ref{equ:mini proa}) follows from the monotonicity. (\ref{equ:mini proc}) follows from Lemma \ref{lem:root} and Lemma \ref{lem:root1}.

Furthermore, we can also obtain the lower bound of $\mathscr{L}({\textbf{\emph{p}}})$ according to the monotonicity of logarithmic function,
\begin{equation}
\mathscr{L}({\textbf{\emph{p}}}) \ge  \log \left( {{p_{\min }}{e^{{1-p_{\min} \over {{p_{\min }}}}}}} \right)=L(p_{\min}).
\end{equation}

Based on the discussions above, we get
\begin{equation}
0 \le \mathscr{L}({\textbf{{p}}}) - L(p_{\min}) \le \log 2.
\end{equation}
\end{proof}
\begin{rem}\label{rem:summax describtion}
Proposition \ref{thm:summax} deduces that the gap between $\mathscr{L}({\textbf{\emph{p}}})$ and $L(p_{\min})$ is less than a constant. In some applications, $p_{\min}\ll1$, so $L(p_{\min}) \gg \log 2$. In this case, one can use $\mathscr{L}({\textbf{\emph{p}}}) \approx L(p_{\min})$ instead.
\end{rem}
\begin{rem}\label{rem:big data key characters}
In Proposition \ref{thm:summax}, it clearly shows the two key characters of big data. Condition $i)$ reflects the rare event finding. Condition $ii)$ reflects the large diversities of events.  This illustrates the validity of definition of NMIM directly in the era of big data.
\end{rem}

\begin{rem}\label{rem:summax applications}
Proposition \ref{thm:summax} explains the sense of NMIM, that the rare events own the majority of information in the viewpoint of message importance, such as seismic record. In this case, we only need to save some small probability events without losing too much information. As a result, it is expected to record large amounts of information with less storage space. Actually, the condition of Proposition \ref{thm:summax} is easy to meet in the scenario of big data.
\end{rem}

In the era of big data, a large amount of data needs to be stored and transferred. Unfortunately, the storage space of the devices and actual channel capacity are always limited. Therefore, it is quite hard to store or transfer all the messages. Based on the idea of dimensionality reduction, we expect to process a small portion of the data which can retain most importance information. In fact, NMIM is exactly the importance measure to do this, because this importance measure magnifies this minority data. The remaining parts of this paper cover specific strategies for solving the storage and transmission problem in big data by means of NMIM.

\section{Compressed Storage by NMIM}
In this section, we discuss the compressed storage problem with NMIM for big data. Consider the following case, where a security camera takes a lot of photographs from time to time. At present, it is easy for us to analyze the images to get the information which we need. For example, a security camera at the door of one supermarket can record whether someone comes in. It will help us to find the thief if burglary happens. However, storing so much data day by day is too hard, and thus it is necessary to compress this image data.

Consider the model shown in Fig.\ref{fig:code model}, in which there are some events and the event with lower probability owns more importance. In order to compress the data for storage, part of the information will be lost in the process of compression, but the degree of loss can be controlled. In this paper, we consider the rule that minimizes the loss of useful information in compressed encoding. Suppose that there is no error in transmitting and the precision of information reconstruction is only related to the coding length in this model. Obviously, the reconstruction error is monotonically decreasing with respect to code length since more detail can be recoded with longer code length. However,  how to measure the loss of information becomes the important question. In fact, the error of important information brings great loss while the common events' loss is insignificant. For example, if the image which contains the suspicious individuals is not clear, we will not able to catch the suspects. However, it will have little influence if the image which shows the usual situation is not clear. In order to describe such kind of things, NMIM is used as the weighting factor. This rule is utilized to minimization of NMIM importance loss. To do so, this optimization problem can be formulated as
\begin{equation}\label{pro:problem1}
\mathop {\min }\limits_{{l_i}} \left\{ {\sum\limits_{i=1}^n {{p_i}{e^{{{1 - {p_i}} \over {{p_i}}}}}{D_f}({l_i})} } \right\},
\end{equation}
where $l_i$ is the encoding length and ${D_f}({l_i}) $ is the reconstruction error function. The source event number is $n$. To solve this optimization problem, the resource condition need to be taken into account. For convenience, we consider the total code length of all the events $K = \sum\limits_i {{l_i}} $ in this paper. We also suppose that the initial code length of every events is assigned as $L$. As a result, the average code length before coding is $L$.

\begin{figure*}
  \centering
  \includegraphics[width=1\textwidth]{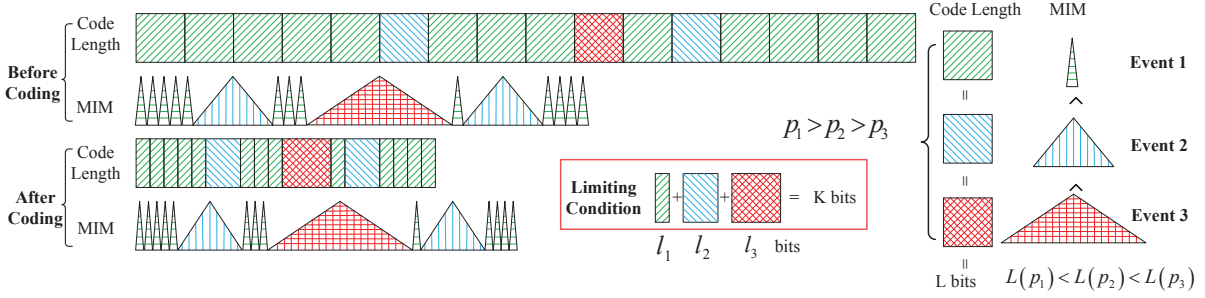} \\
  \caption{The MIM coding model.}\label{fig:code model}
\end{figure*}

\subsection{Reciprocal Error Model}
In this section, the reciprocal error model is adopted, in which we minimize the expected values of NMIM with limited average encoding length. Problem in (\ref{pro:problem1}) can be represented by
\begin{flalign}\label{equ:re code model}
\mathcal{P}_1: \,\, \mathop {\min }\limits_{{l_i}} \,\,\, &\sum\limits_{i=1}^n {{p_i}{{e^{{1-p_i} \over {p_i}}} {{l_i}^{-1}}}}  \\
\textrm{s.t.}\,\,\,& K = \sum\limits_{i=1}^n {{l_i}} \tag{\theequation a}\label{equ:re code model a}\\
&\text{if}~l_i\le L, \quad i=1,2,...,n   \tag{\theequation b}\label{equ:re code model b}
\end{flalign}

\begin{prop}\label{thm:code1}
For a given source event number $n$ and the sum of every coding length K ($K \le L$), the solution to the problem in (\ref{equ:re code model}) is given by
 \begin{equation}\label{equ: prop code1}
  {l_i} = {\left\lceil {{{\sqrt {M({p_i})} } \over {\sum\limits_{i=1}^n { {\sqrt {M({p_i})} }  } }}K} \right\rceil  }.
\end{equation}
${\left\lceil x\right\rceil}$ is the smallest integer larger than or equal to $x$.
\end{prop}

\begin{proof}
Refer to the Appendix \ref{Appendices C}.
\end{proof}

When $p_{min} \ll 1$ or $n\gg1$, the compression ratio can be calculated by
 \begin{equation}
\begin{split}
   {{\sum\limits_{i = 1}^n {{p_i}{l_i}} } \over L} &= {{\sum\limits_{i = 1}^n {{p_i}\sqrt {M\left( {{p_i}} \right)} } } \over {\sum\limits_{i = 1}^n {\sqrt {M\left( {{p_i}} \right)} } }}{K \over L} \overset{(c)}\approx {p_{\min }}{K \over L},
  \end{split}
\end{equation}
because ${\sum\limits_{i = 1}^n {\sqrt {M\left( {{p_i}} \right)} } } \approx {\sqrt {M\left( {{p_{\min }}} \right)} }$ and ${\sum\limits_{i = 1}^n {{p_i}\sqrt {M\left( {{p_i}} \right)} } } \approx{{p_{\min }}\sqrt {M\left( {{p_{\min }}} \right)} }$ when $p_{min} \ll 1$ or $n\gg1$, which can be derived easily in the same way as Remark \ref{rem:summax describtion}. Generally speaking, the gap between $K$ and $L$ is usually small and $p_{min} \ll 1$. In this case, ${p_{\min }}{K \over L} \ll 1$, which means the data is compressed greatly.

\subsection{Exponent Error Model}
In the following, we assume the relation between encoding error and codeword lengths is exponential function. That is to say, exponent error model is used as the reconstruction error function and NMIM is taken as wight to measure the encoding loss. Problem in (\ref{pro:problem1}) can be represented by

\begin{flalign}\label{equ:ex code model}
\mathcal{P}_2: \,\,\mathop {\min }\limits_{{l_i}}\,\,\, & {\sum\limits_{i = 1}^n {{p_i}} {e^{{{1 - {p_i}} \over {{p_i}}}}}{\gamma ^{ - {l_i}}}}  \\
\textrm{s.t.}\,\,\,& K = \sum\limits_{i}^n {{l_i}} \tag{\theequation a}\label{equ:ex code model a}\\
& \text{if}~ l_i\le L, \quad i=1,2,...,n   \tag{\theequation b}\label{equ:ex code model b}
\end{flalign}

\begin{prop} \label{thm:code2}
For given source event number $n$ and the required encoding length $K$ ($K\le L$), the error probability for codeword lengths $l_i$ is ${\gamma ^{ - {l_i}}}$ where $\gamma$ is size of code alphabet. The solution to the problem $\mathcal{P}_2$ is given by
 \begin{equation}\label{equ: prop code2}
{l_i} = {\left\lceil {{{\left( {\ln M({p_i}) - \sum\limits_{i = 1}^{\tilde N} {\ln M({{\tilde p}_i})} /\tilde N + K\ln \lambda /\tilde N} \right)} \mathord{\left/
 {\vphantom {{\left( {\ln M({p_i}) - \sum\limits_{i = 1}^{\tilde N} {\ln M({{\tilde p}_i})} /\tilde N + K\ln \lambda /\tilde N} \right)} {\ln \lambda }}} \right.
 \kern-\nulldelimiterspace} {\ln \gamma }}} \right\rceil ^ + }
\end{equation}
where ${\tilde N}$ is the positive number of all the $l_i,i=1,2,...,n$ and $\{\tilde p_i, i=1,2,...,\tilde N\}$ is part of the commutative sequence of $\{ p_i, i=1,2,...,N\}$ in increment order which satisfies ${{\ln M\left( {{p_i}} \right) + \ln \ln \gamma  - \ln \left( { - \lambda } \right)} }>0$. ${\left\lceil x\right\rceil}$ is the smallest integer larger than or equal to $x$ and
\begin{equation}
\begin{split}
 (x)^+=\left\{
   \begin{aligned}
 & x, \quad x>0 \\
 & 0,\quad x\leq0 \\
   \end{aligned}
   \right.
   \end{split}
\end{equation}
\end{prop}

\begin{proof}
Refer to the Appendix \ref{Appendices D}.

\end{proof}

When $\tilde N=n$, the corresponding compression ratio is
\begin{flalign}\label{equ:ex compress rate}
 {{\sum\limits_{i = 1}^n {{p_i}{l_i}} } \over L} &= {{\sum\limits_{i = 1}^n {{p_i}\ln M({p_i}) - \sum\limits_{i = 1}^n {\ln M({p_i})} /n + K\ln \gamma /n} } \over {\ln \gamma L}} \le {K \over {nL}}  
\end{flalign}
because ${\sum\limits_{i = 1}^n {{(p_i)}\ln M\left( {{p_i}} \right)}  }-{\sum\limits_{i = 1}^n {\ln M\left( {{p_i}} \right)}/n  }\le0$ according to (\ref{equ:G and B Center}). In general, we select $K$ being less than or equal to $L$ and ${n}\gg1$, so ${K /({nL}}) \ll1$, which means this coding scheme does achieve significant compression of data.

In fact, $\ln(x)$ is a more smoothing operator than $\sqrt{x}$. Therefore, the encoding length is more even in exponent error model than that in reciprocal model.

These two propositions focus on the case where $K \le L$, which guarantees $l_i \le L$. In fact, when $K>L$, the above methods can still be applied with small changes to them. If $l_i>L$, let $l_i \leftarrow L$ and $K\leftarrow K-L$, and then execute the above compressed algorithm literally, iterate it until all code length are smaller than $L$. For convenience, the program of our new compressed storage strategy is summarized in Algorithm \ref{alg:code}.

\begin{algorithm}[htb]
\caption{Compressed Coding}
\label{alg:code}
\begin{algorithmic}[1]
\REQUIRE ~~\\
The probability distribution of source, $\textbf{p}=\{p_i, i=1, 2, ..., n\}$\\
The sum of each encoding length, $K$\\
The initial code length, $L$
\ENSURE ~~\\
The compressed code length, $l_i,i=1,2,...,n$
\STATE    ${l_i}^\prime  \leftarrow f(\textbf{p},K,L)$  \quad\quad\quad\quad  $\vartriangleright$ See (\ref{equ: prop code1}) or (\ref{equ: prop code2})
\STATE    Sort ${l_i}^\prime $ and find the maximum one $l_{j}^\prime$
\STATE    \textbf{if}  $l_{j}^\prime>L$
\STATE    $l_{j}=L$, $K \leftarrow K-L$, $\textbf{p} \leftarrow \{p_i, i=1, 2, ..., n, i \ne j\}$ \\
               and go back to step 1
\STATE    \textbf{else} ${l_i}\leftarrow  {l_i}^\prime $ and go to step 6
\RETURN $l_i$
\end{algorithmic}
\end{algorithm}

In Algorithm \ref{alg:code}, the code length depends on the event's importance $M(p_i)$, in which more code size is assigned to the important events and little code size is assigned to common events. Thus it is feasible to use a small part of storage while retaining most important information.

\section{Analysis of Transmission by NMIM}
We discuss the message transmission problem in terms of NMIM in this section. As a message is transmitted from one side to another, people generally prefer to see the important part of message rather than the whole message itself, and NMIM will play a key role in such kinds of applications. In this case, we focus on the maximum transmission problem with limited NMIM loss under the physical environment constants, such as channel capacity, the transmission delay, etc., which helps to improve the design of communication systems.

\subsection{The Change of NMIM}
First of all, we study the change of NMIM during the transmission of the message.
 \begin{defn}\label{Rate Distortion}
\par The change of NMIM is defined as
\begin{equation}
\psi(X,Y)   ={ \mathscr{L}({\textbf{p}_x}) -  \mathscr{L} ({\textbf{p}_{y}})} .
\end{equation}
where $\textbf{p}_x$ and $\textbf{p}_y$ are the probability distributions of random variables $X$ and $Y$.
\end{defn}
 In fact, the value of $\psi(X,Y)$ could be positive or negative. If the value of $\psi(X,Y)$ is positive, it means that there is importance loss during the transmission, such as the lose of crucial data. On the other hand, the loss of importance will be over-interpreted if it is negative, which results in waste of resources. In general, the importance loss produces more severe impairment than importance over-interpreted. In the following part, we only consider importance loss for convenience.

According to Proposition \ref{thm:summax}, if $n \gg 1$ or $p_{xm}\ll \frac{1}{1+\ln(n-1)}$, $p_{ym}\ll \frac{1}{1+\ln(n-1)}$, we have
\begin{flalign}\label{equ:low and upper bound for L}
&\psi(X,Y) \ge L(p_{xm})-L(p_{ym})-\log2 \\
&\psi(X,Y) \le L(p_{xm})-L(p_{ym})+\log2  \tag{\theequation a}\label{equ:low and upper bound for L a}
\end{flalign}
where $p_{xm}$ and $p_{ym}$ is the unique minimum probability of $\textbf{p}_x$ and $\textbf{p}_y$, respectively, and it is the loose upper and lower bounds for $\psi(X,Y)$.

Actually, in most cases, $\psi(X,Y)$ is hard to study. As a result, some special cases are taken into account here. We first select Binary Symmetric Channel (BSC) since it captures most of the characteristic of the general problem as a traceable model of channel with errors \cite{Elements}, while Bernoulli($p$) source ($0<p\le0.5$) is adopted for its simpleness and representativeness.

When a message is transmitted from one side to another, it should not change too much for an effective transmission, otherwise this transport fails. Based on this idea, it is reasonable to assume $ \varepsilon (1 - 2p)\ll p$. In most of sceneries, $p$ is not close to $0.5$, so $ \varepsilon \ll p$.
\begin{figure}
  \centering
  \includegraphics[width=0.35\textwidth]{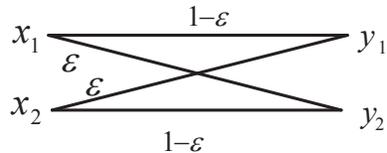} \\
  \caption{Binary symmetric channel.}\label{fig:channelBSC}
\end{figure}

\begin{prop}\label{thm: BSC}
For a given BSC and Bernoulli(p) source, if the transmission error probability $\varepsilon $ is relatively small ($\varepsilon \ll p$), the change of NMIM is
\begin{equation}
\psi(X,Y) \approx  \varepsilon s(p)
\end{equation}
where $s(p)={{(1 - p)(1 - 2p)} \over {{p^2}}} $.
\end{prop}

\begin{proof}
Refer to the Appendix \ref{Appendices E}.
\end{proof}

\begin{rem}\label{rem:BSC1}
This proposition means that the change of NMIM is proportional to the $\varepsilon$ and $s(p)$. $s(p)$ is a rational fraction with respect to parameter $p$.
\end{rem}

Moreover, when $p \ll 1$, $s(p)$ can be simplified to $\frac{1}{p^2}$. In this case, The importance loss is proportional to the $\varepsilon$ and inversely proportional to $p^2$.

In fact, this result means that the impact of channel and that of source can be decouple under certain condition. On the one hand, $p$ is the attribute of source and has nothing to do with channel, so $s(p)$ is the representation of source. Obviously, $s(p)$ decreases with the increasing of $p$, so a little loss of information leads to heavy loss when $p$ is very small (the message itself is important in this case). On the other hand, $\varepsilon$ describes the channel performance, such as the channel capacity (the channel capacity of BSC is $1-H(\varepsilon)$ \cite{Elements}). If channel capacity is relatively small, which means that $\varepsilon$ is big, the importance loss will be large. Actually, in the actual communication process, the source can not be chosen and we can only improve the channel performance. In order to improve communication quality in BSC, what we can do is make $\varepsilon$ smaller, which is consistent with the known result in information theory.

\subsection{NMIM Loss Distortion}\label{sec:MIM Loss}
In this section, we still consider the Bernoulli($p$) source, but the general binary channel is analyzed rather than BSC. In fact, this special transmission process can also be seen as source coding. The basic problem in NMIM loss distortion is what the maximum expected NMIM loss is as given a particular distortion. Here the NMIM loss distortion function is defined as follows.
 \begin{defn}\label{Rate Distortion}
\par The NMIM loss distortion function $R^{(MIM)}(D)$ for a source $X$ with distortion measure $d(x,{\hat x})$ is defined as
\begin{equation}
{R^{(MIM)}}(D) = \mathop {\max }\limits_{p\left( {\hat x|x} \right):\sum\nolimits_{(x,\hat x)} {p(x)p(\hat x|x)d(x,\hat x) \le D} } \left\{ { \mathscr{L}({\textbf{p}_x}) -  \mathscr{L} ({\textbf{p}_{\hat x}})} \right\}.
\end{equation}
\end{defn}

This function describes the upper bounds of changes of NMIM with distortion. Obviously, the NMIM loss distortion function  $R^{(MIM)}(D)$ is the supremum of the loss of NMIM for a given distortion D.
${ \mathscr{L}({\textbf{p}_x}) -  \mathscr{L} ({\textbf{p}_{\hat x}})}$ describes the changes of NMIM before and after coding. In the best case, this difference value is zero if $\mathscr{L}({\textbf{p}_x}) =\mathscr{L} ({\textbf{p}_{\hat x}})$, which means no NMIM loss in coding process. In fact, NMIM invariably changes more or less due to distortion.

In fact, there is a tradeoff between NMIM loss and distortion. We would like calculate the NMIM loss distortion function for some simple sources to explain its physical meaning. Let $R^{(MIM)}(D)$ describe NMIM loss of Bernoulli($p$) source and the distortion is less than or equal to $D$.

\begin{prop}\label{thm:Rate MIM LOSS}
(NMIM-loss-distortion region). When Hamming distortion measure D ($D\ge0$) is used, the NMIM loss distortion function for a Bernoulli(p) source is given by
\begin{flalign}
&{R^{(MIM)}}(D)= 
 & \left\{
   \begin{aligned}
 &\log{e^2}+\delta(p)-\log((p + D){e^{{1 \over {p + D}}}} + (1 - p - D){e^{{1 \over {1 - p - D}}}}), p+D<{1\over2}, 0\le D \le 1-p \\
 &\log{e^2}+\delta(p)-\log((p - D){e^{{1 \over {p - D}}}} + (1 - p + D){e^{{1 \over {1 - p + D}}}}), p-D>{1\over2}, 0\le D \le p  \\
 &\delta(p), else,  
   \end{aligned}
   \right.
\end{flalign}
where $\delta(p)=\log({p{e^{{1 \over p}}} + (1 - p){e^{{1 \over {1 - p}}}}})-\log{e^2}$.
\end{prop}
\begin{proof}
Refer to the Appendix \ref{Appendices F}.
\end{proof}

Fig.\ref{fig:rate1} illustrates NMIM loss distortion function. The figure shows that the NMIM loss distortion increases with the increasing of $D$ when $D<0.5-p$ and it stays the same when $D \ge 0.5-p$, and thus the turning point of every line is at the point $(0.5-p,\delta(p))$. From Fig.\ref{fig:rate1}, it is also observed that a small distortion will cause a big loss of NMIM if $p$ is very small.
\begin{figure}
  \centering
  \includegraphics[width=0.45\textwidth]{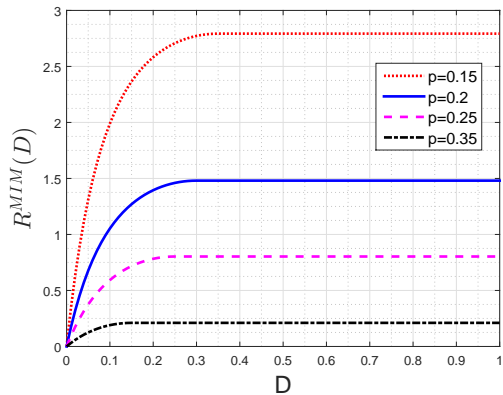} \\
  \caption{NMIM loss distortion function for a Bernoulli(p) source. The base of the logarithm is $e$ in NMIM.}\label{fig:rate1}
\end{figure}

\subsection{Maximum transmission}
When a message is transmitted from one side to another, there are double layers of characteristics for it, where the first layer describes the amount of information, and the second layer performs importance marking. The amount of information describes how much resources we need to expend for storing or transmitting and it is usually measured by Shannon entropy in the actual communication. The importance determines what the cost is once transmission errors appear. Usually, only important part of message is effective and people do not need the whole message itself. Therefore, NMIM is the proper importance measure in such kinds of applications.

Here we focus on the transmission problem to maximize the received entropy with limited NMIM loss and the channel capacity. This model is summarized as follows. For a transmission process $X \to Y$, the entropy in transmitting terminal $H(X)$ and transmission time $t$ determine the bit rate. In order to distortionless transmission, we have $H(X)/t \le C$ where $C$ is channel capacity in the required transmission time, which shows the limit on transmission due to channel itself. Furthermore, human beings expect importance loss as small as possible, so people usually prescribe a limit to importance loss, such as $0 \le {\psi(X,Y)}  \le \Delta$, where $\Delta$ is the maximum allowable change of NMIM. In fact, this condition ensures that the receiver can afford the loss. Within this region, people would like to maximize entropy in receiving terminal, which means receiver can get as much effective information as possible . The mathematical expression for this model is given by
\begin{flalign}\label{equ:HC1 optization}
\mathcal{P}_3: \,\,\mathop {\max }\limits_Y\,\,\, &H(Y)  \\
\textrm{s.t.}\,\,\,& 0\le {\psi(X,Y)}  \le \Delta \tag{\theequation a}\label{equ:HC1 optization a}\\
& H(X)/t\le C  \tag{\theequation b}\label{equ:HC1 optization b}
\end{flalign}

Unfortunately, it is not a convex optimization problem, so the solution to $\mathcal{P}_3$ is hard to get in general. In order to facilitate the analysis, we consider this problem in the extreme case firstly. When $\Delta=0$, which means that there is no NMIM change before and after transmission, the message $X$ is transmitted to the receiving terminal perfectly. As a result, $\mathop {\max }\limits_Y H(Y)=\mathop {\max }\limits_X H(X)=Ct$. On the other hand, when $\Delta \ge \mathscr{L}({\textbf{p}_x})+ (n-1)$, the distribution of $Y$ can be uniform distribution according to (\ref{equ:lowerbound}). At this time, $\mathop {\max }\limits_Y H(Y) =\log_2 n$. However, when $0<\Delta<\mathscr{L}({\textbf{p}_x})+ (n-1)$, the solution is fairly complicated.

For simplifying the analysis, we consider this maximum transmission problem in Bernoulli($p$) source as a typical example.

\begin{prop}\label{thm:HC1}
For a Bernoulli($p$) source, if $\Delta\ge0$ and $0<Ct \le 1$ bits, the solution to the constrained maximization $\mathcal{P}_3$ is given by
\begin{equation}
\begin{split}
 \mathop {\max }\limits_Y H(Y)=
  \left\{
   \begin{aligned}
   &Ct,\Delta=0\\
   &H(p_0+D^{(MIM)}(\Delta)),0<\Delta \le \delta(p_0)\\
   &1,\Delta>\delta(p_0)
   \end{aligned}
   \right. ,
   \end{split}
\end{equation}
where $H(p_0)=Ct$ and $\delta(p_0)$ is defined in (\ref{equ:delta_p0}). $D^{(MIM)}(\cdot)$ is the inverse function of $R^{(MIM)}(D)$.
\end{prop}

\begin{proof}
Refer to the Appendix \ref{Appendices G}.
\end{proof}

\begin{rem}\label{rem:HC1}
Proposition \ref{thm:HC1} shows that there are growth region and saturation region for the maximum entropy of receiving terminal with respect to $\Delta$. The turning point is at $\Delta=\delta(p_0)$.
\end{rem}
Remark \ref{rem:HC1} shows that people can not increase $\Delta$ unlimitedly to obtain better performance when other parameter fixed. $\delta(p_0)$ is the minimum allowable NMIM loss when the receiver wants to maximize the entropy of the information.

 Moreover, according to $H(p_0)=Ct$, so $p_0=H^{-1}(Ct)$ when we consider $0 \le p_0 \le 0.5$. Refer to \cite{Elements}, $H^{-1}(\cdot)$ is monotonic increasing function, which means $p_0$ increases with increasing of $Ct$. In fact, the monotonicity of $\delta(\cdot)$ is similar with $L(\cdot)$ which is monotonic decreasing in $(0,0.5]$ according to \ref{sec:monotonicity}. As a result, $\delta(p_0)$ decreases with increasing of $Ct$, which means the receiving amount of information increasing with increasing of the product of $C$ and $t$ when the allowed importance loss is fixed. It also shows that the channel capacity and transmission time affect the performance of maximum transmission together. In fact, $Ct$ is mobile service, which reflects the cumulative transmitted information for a channel within a period of time \cite{dong2012deterministic}.

Based on the result discussion above, NMIM can be used as a very useful importance measure to help us analyze and design practical communication system.

\section{Numerical Results}
In this section, numerical results will be presented to validate the above results in this paper.

\subsection{The properties of minimum probability}
First of all, the relationship between NMIM of total events and that of minimum probability event is discussed when the probability distribution is Zipf, Normal and Rayleigh distribution. The probability of Zipf distribution is $P\{X=k\}={Z}/{k^{1.01}},k=1,2,...,n$ and $Z={[\sum_{k=1}^n {(1/k)}^{1.01}]}^{-1}$. The probability of Normal distribution is $P\{X=k\}=(1/{\sqrt{2\pi \sigma^2}})e^{-{(k-\mu)}^2/{2 \sigma^2}},k=1,2,...,n-1$ where $\mu=0.51 \times n$ and $\sigma^2=n$ and $P\{X=n\}=1-\sum_{k=1}^{n-1} P\{X=k\}$ for normalization. The probability of Rayleigh distribution is $P\{X=k\}=(k/b^2)e^{-k^2/(2b^2)},k=1,2,...,n-1$ where $b=n/2.5$ and $P\{X=n\}=1-\sum_{k=1}^{n-1} P\{X=k\}$ for normalization. The range of event number $n$ is selected from $5$ to $20$. For these three typical distributions, it is easy to find that all the minimum probabilities are unique and far less than ${1 \over {1 + \ln (n - 1)}}$, which means that the condition of Proposition \ref{thm:summax} is met.

Fig. \ref{fig:MIM_minmum1} shows the gap between $\mathscr{L}({\textbf{\emph{p}}})$ and $L(p_{\min})$ is less than a constant ($\log2$), which proves validity of Proposition \ref{thm:summax}. Moreover, $\mathscr{L}({\textbf{\emph{p}}}) -L(p_{\min})$ decreases with increasing of $n$. In fact, this gap will almost disappear in comparison to $L(p_{\min})$ when $n$ is relatively large. At this time, we can use $\mathscr{L}({\textbf{\emph{p}}}) \approx L(p_{\min})$, which is shown in Fig. \ref{fig:MIM_minmum}.

\begin{figure}

\begin{minipage}{0.48\linewidth}
  \centerline{\includegraphics[width=8.0cm]{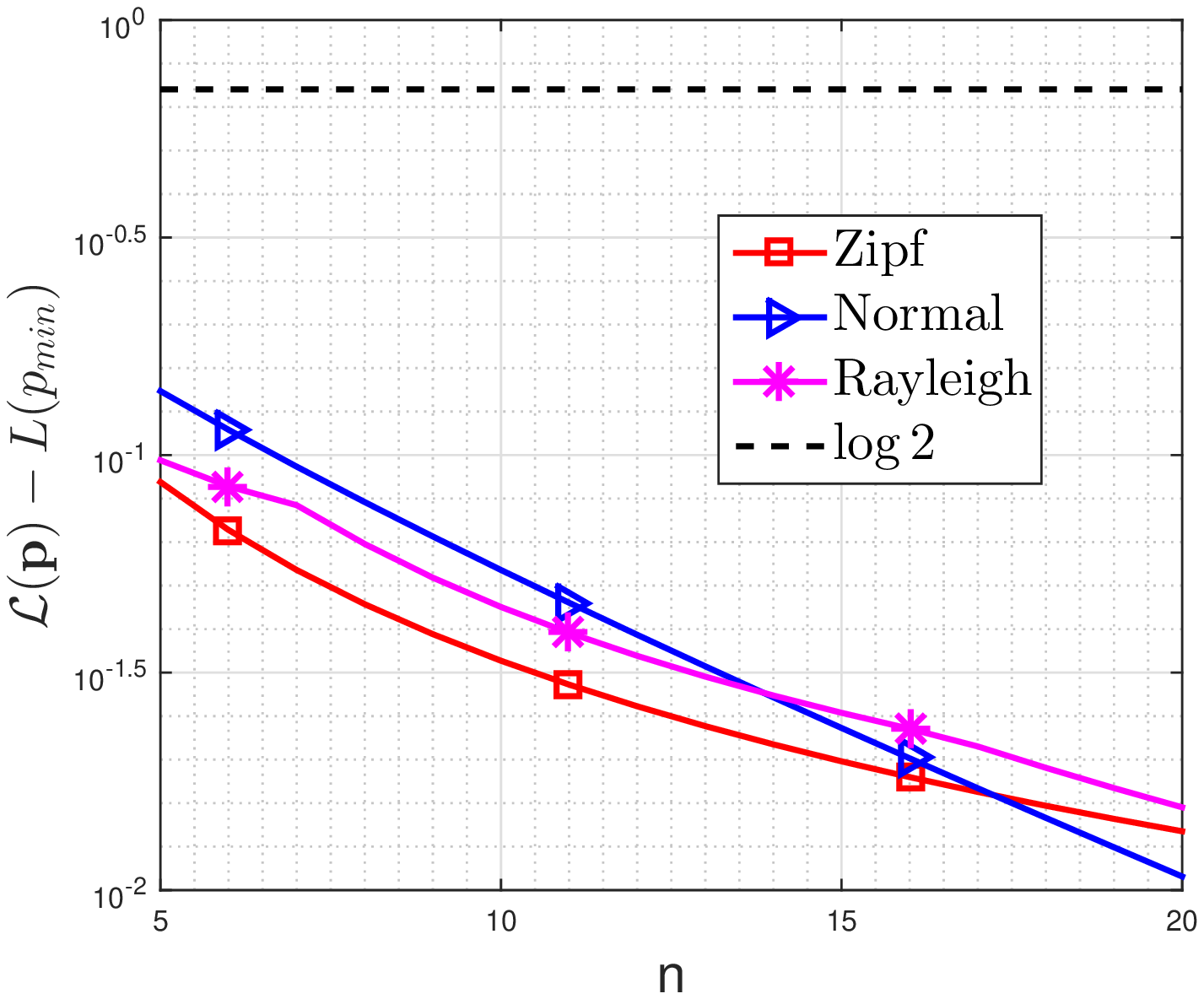}}
  \caption{$\mathscr{L}({\textbf{\emph{p}}})-L(p_{\min})$ vs. n. (${\mathcal{L}}(P)$ in axis Y stands for $\mathscr{L}({\textbf{\emph{p}}})$)}\label{fig:MIM_minmum1}
\end{minipage}
\hfill
\begin{minipage}{0.48\linewidth}
  \centerline{\includegraphics[width=8.0cm]{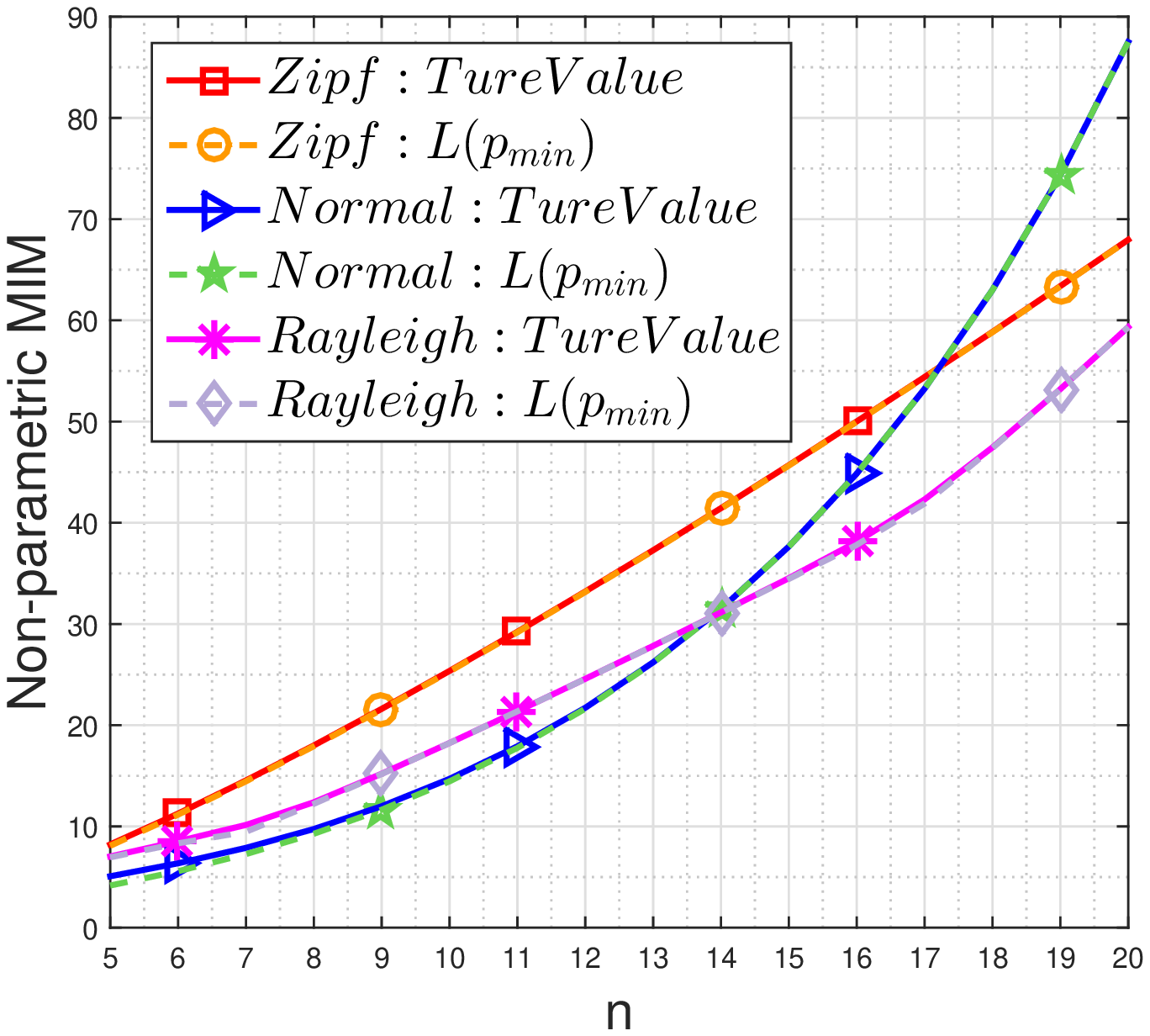}}
  \caption{NMIM vs. n.}\label{fig:MIM_minmum}
\end{minipage}
\end{figure}



\subsection{Compressed coding by NMIM}
Next, we focus on conducting the simulation by computer to compare and analyze the compressed coding. The initial code length of every events $L$ is $100$. The required coding length $K$ is varying from $10$ to $200$. The probability of events ${\textbf{p}}=(0.010,0.215,0.037,0.292,0.446)$. The size of alphabet for exponent $\gamma$ is $2$. To better show the performance, some other coding schemes are listed here. The code length is distributed equally in Code 1. The code length in Code 2 is assigned to each events according to its possession rate in total probability. More code size is assigned to the one with smaller probabilities in order to reduce the importance loss.

Fig. \ref{fig:MIMcode1} and Fig. \ref{fig:MIMcode2} respectively show performance of message importance loss in the reciprocal and exponent error model. Some observations are obtained. The constraints on the average code length is true. That is, each event's average code length is smaller than the uncompressed value. More important, the average code length of NMIM code is the smallest in most time among the three considered encoding strategies. The loss of message importance decreases with increasing of $K$ for all the compressed encoding strategies. There exists a threshold $K_0$ ($K_0=100$ in Fig. \ref{fig:MIMcode1} and $K_0=140$ in Fig. \ref{fig:MIMcode2}). When $K>K_0$, the message importance loss for NMIM code will almost disappear. In general, there is a tradeoff between average code length and accuracy for any compressed coding, but NMIM provides a new compressed way taking events importance into account.

There are also some difference between Fig. \ref{fig:MIMcode1} and Fig. \ref{fig:MIMcode2}. In Fig. \ref{fig:MIMcode1}, the message importance loss decreases slowly when $K<90$ and decreases rapidly when $K>90$. However, in Fig. \ref{fig:MIMcode2}, the loss of message importance reduces fast from the beginning and then slows down when $K>140$. The average code length in Fig. \ref{fig:MIMcode2} is bigger than that in Fig. \ref{fig:MIMcode1} when $K>140$.

\begin{figure}

\begin{minipage}{0.48\linewidth}
  \centerline{\includegraphics[width=8.0cm]{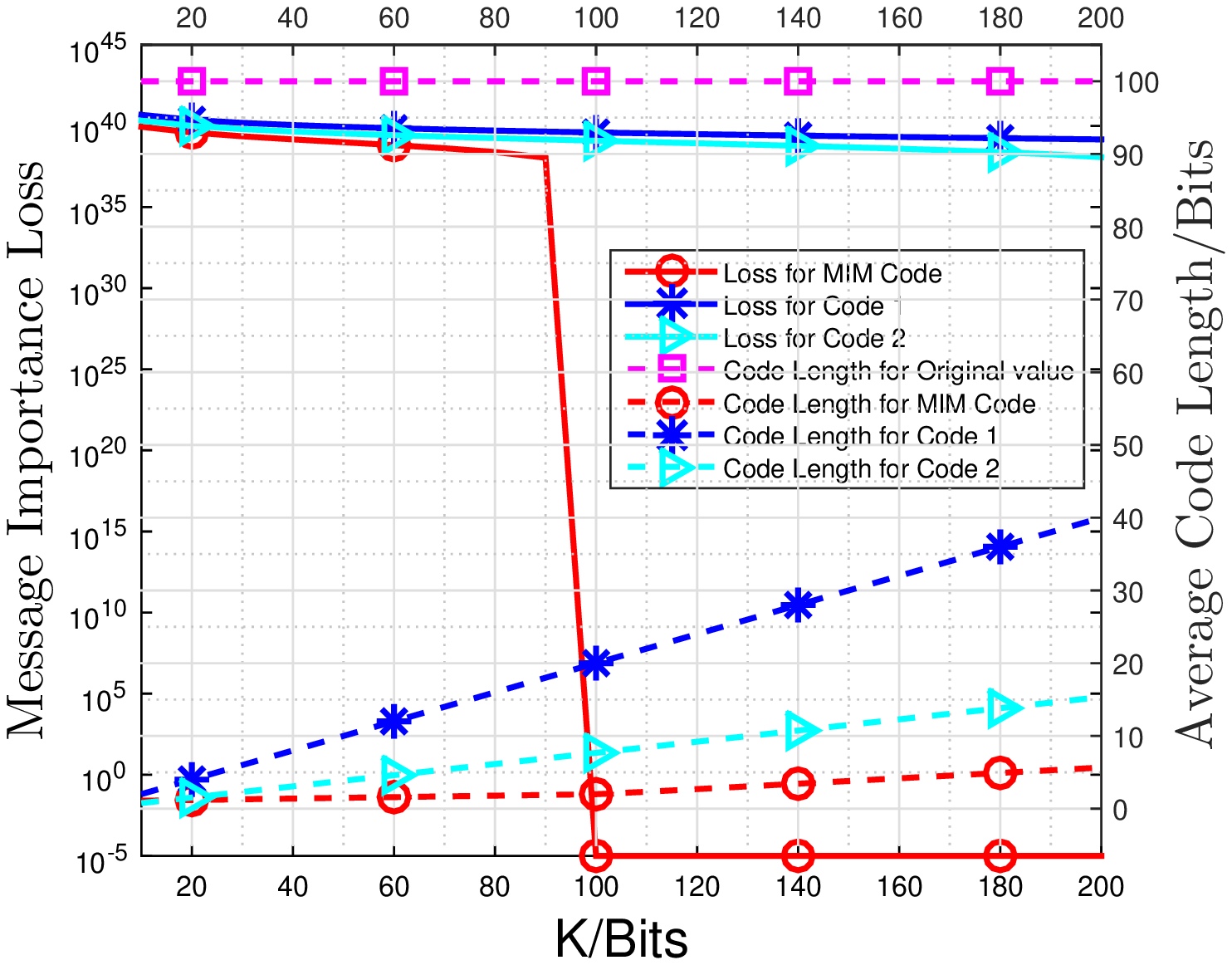}}
  \caption{The performance of importance loss and code length in reciprocal error model.}\label{fig:MIMcode1}
\end{minipage}
\hfill
\begin{minipage}{0.48\linewidth}
  \centerline{\includegraphics[width=8.0cm]{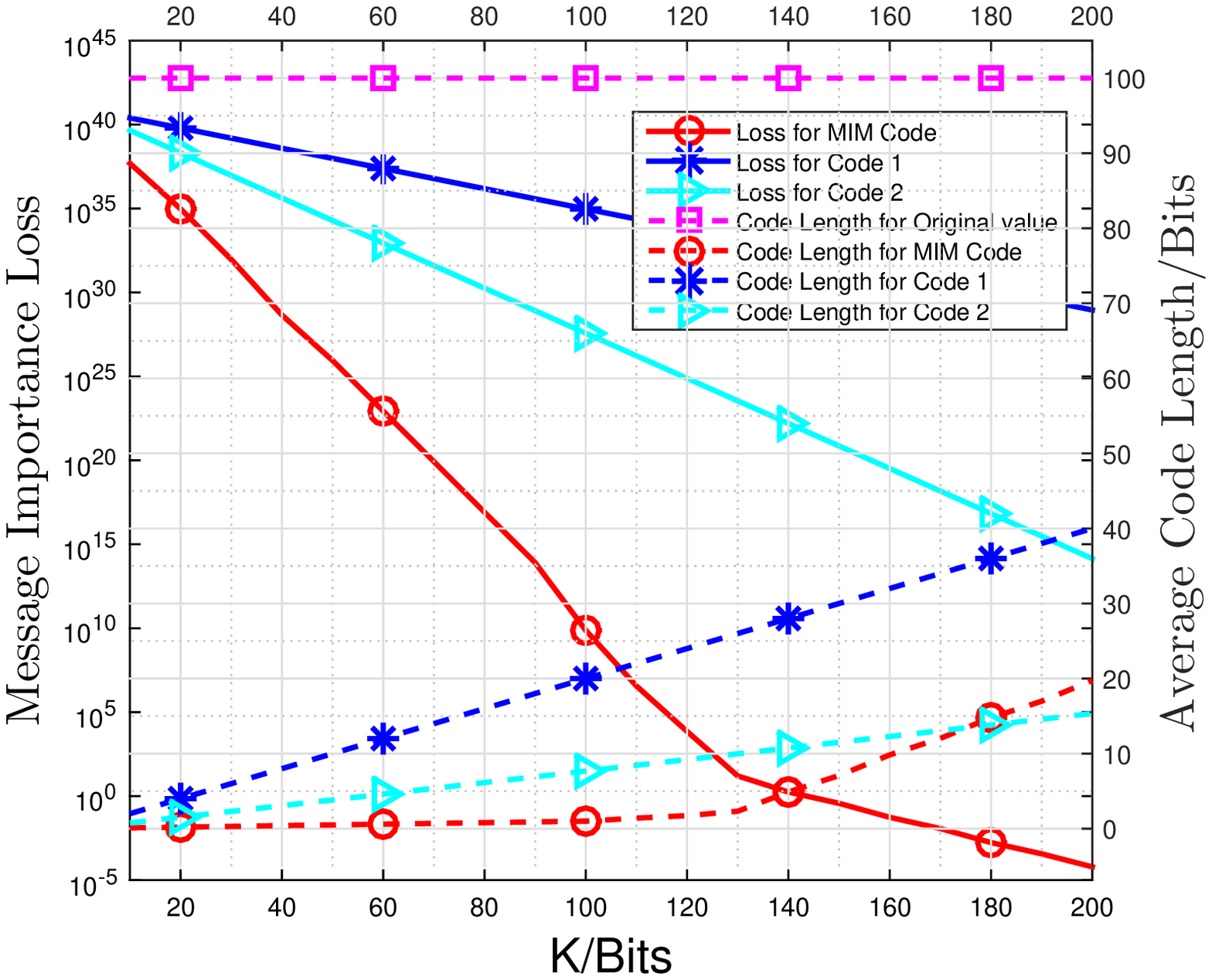}}
  \caption{The performance of importance loss and code length in exponent error model.}\label{fig:MIMcode2}
\end{minipage}
\end{figure}



\subsection{Transmission by NMIM}
Fig. \ref{fig:channel analysis BSC} shows the change of NMIM in BSC with $\varepsilon=0.01$. The approximate value is defined in (\ref{equ:approx MIM change}). The difference between approximate value and $\varepsilon s(p)$ is little, which proves the validness of (\ref{equ:BSC_proof b}). The change of NMIM decreases with increasing of probability $p$. It is noted that $\psi(x,y)$ is zero when $p=0.5$. At this time, the input and output of BSC are both uniform distribution. Obviously, the importance will not change in this case. Furthermore, the approximate value and $\varepsilon s(p)$ are very close to the true value. As a result, $\varepsilon s(p)$ is a good approximation in this case. In the range of the error permitted, ${\varepsilon}/{p^2}$ is also a feasible approximate value. In this figure, we only consider the case in which $0 < p < 0.5$ for the fact that Bernoulli(p) is symmetry about $p = 0.5$.

\begin{figure}
\begin{minipage}{0.48\linewidth}
  \centerline{\includegraphics[width=8.0cm]{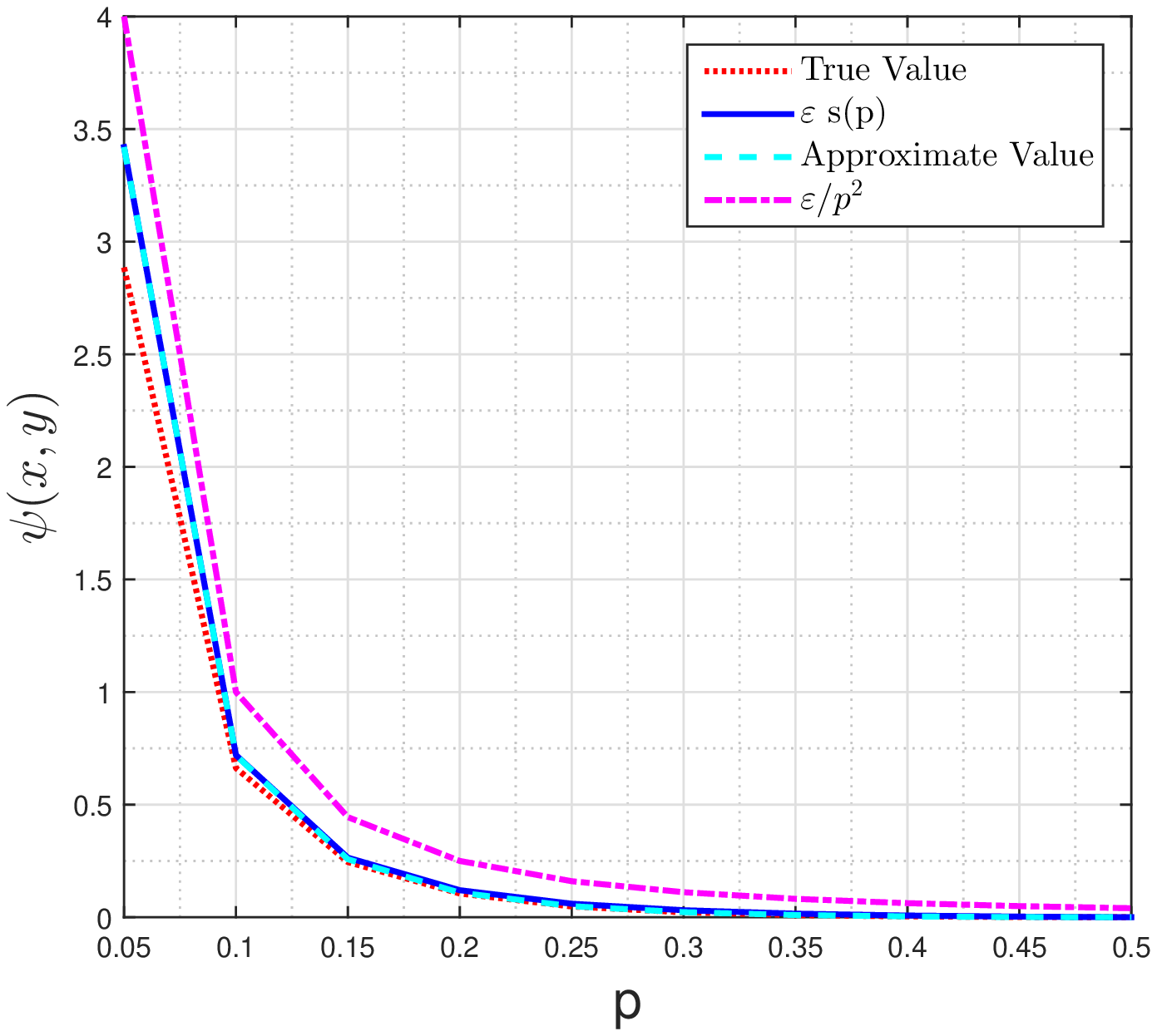}}
  \caption{$\psi(x,y)$ vs. $p$ when $\varepsilon=0.01$ in BSC.}\label{fig:channel analysis BSC}
\end{minipage}
\hfill
\begin{minipage}{0.48\linewidth}
  \centerline{\includegraphics[width=8.0cm]{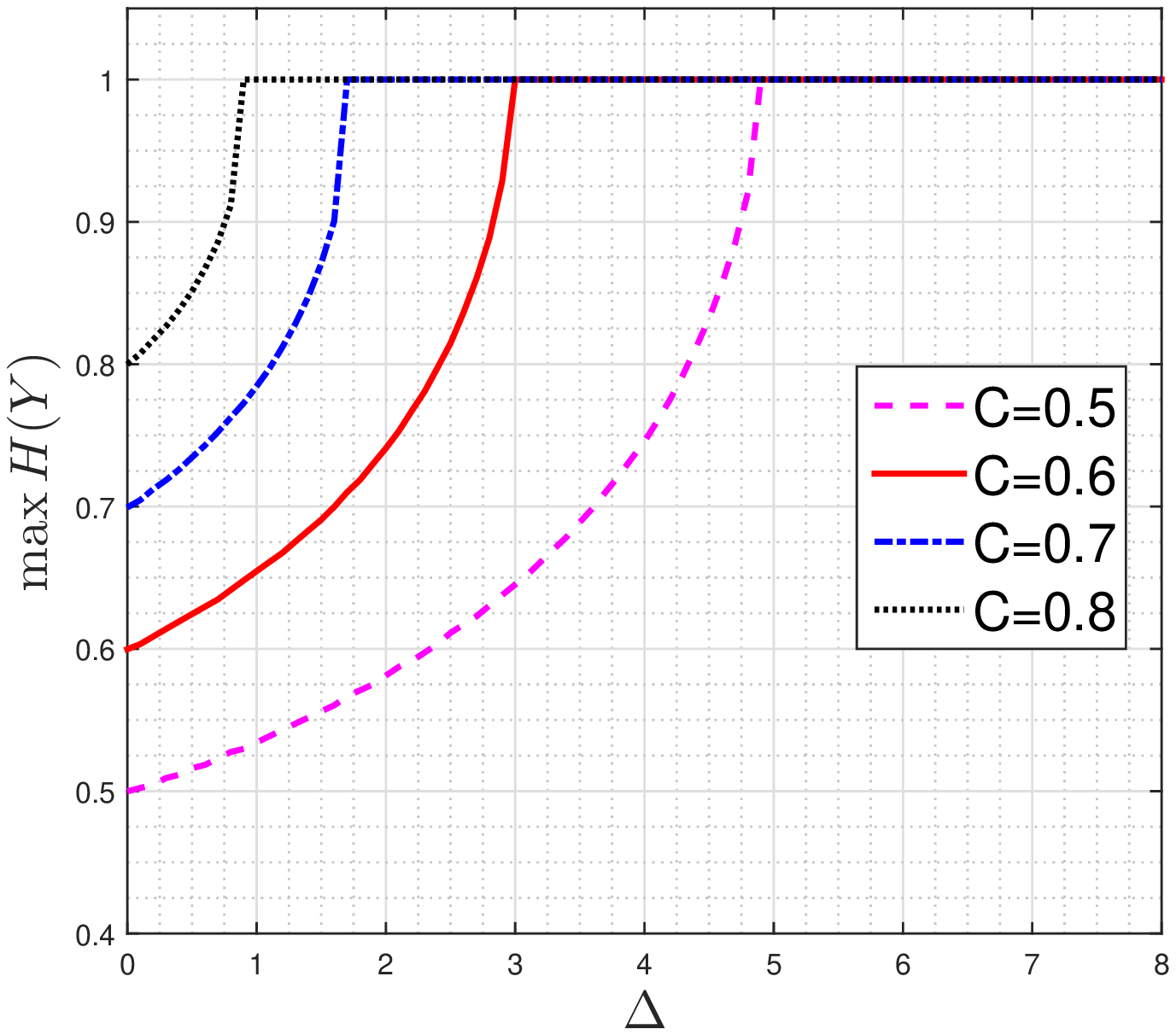}}
 \caption{$\mathop {\max }\limits_Y H(Y)$ vs. $\Delta$ in Bernoulli($p$) source when $t=1$. The base of the logarithm is $e$ in MIM and we take all logarithms to base $2$ in entropy.}\label{fig:HC1}
\end{minipage}
\end{figure}

Fig. \ref{fig:HC1} shows the relationship between the maximum entropy in the receiving terminal and the maximum allowable change of NMIM in Bernoulli($p$) source, which confirms Proposition \ref{thm:HC1}. The transmission time $t$ is $1$. For a given channel capacity, the maximum entropy in the receiving terminal increases from $C$ to $1$ when $0\le \Delta \le \delta(p_0)$ ($H(p_0)=C$). When $\Delta>\delta(p_0)$, the maximum entropy remains unchanged and the value is $1$. For the same $\Delta$, the maximum entropy increases with increasing of $C$ before it reaches saturation ($\max H(Y)=1$). In general, there are growth region ($0,\delta(p_0)$) and saturation region ($\delta(p_0),\infty$) for the maximum entropy, and the length of growth region decreases with increasing of channel capacity.

\section{Conclusion}
In this paper, we discussed the problem of storage and transmission in big data with taking message importance into account. As an extension of MIM, we defined NMIM to measure the message importance, in which we mainly focus on the minority subset different from Shannon entropy. By fully analyzing the properties of NMIM, we found that the difference between message importance of low-probability events and that of all the events is less than a constant, if either the low-probability value is extremely small or the events number to be considered is large. These two conditions are correspond to the two major characters of big data, which are the rare event finding and the large diversities of events.

Then we proposed an effective compressed storage strategy based on NMIM, which coincides with the cognitive mechanism of human beings, and greatly reduces the average code length with little importance loss. We further analyzed the change of importance during transmission and defined NMIM loss distortion function to characterize the tradeoff between NMIM loss and distortion. Importantly, we found that there are growth region and saturation region for maximum transmission problem, which helps to improve the design of practical communication systems.


\appendices
\section{Proof of Lemma \ref{lem:root}}\label{Appendices A}
\begin{proof}
\begin{equation}\label{equ:min}
\begin{split}
{{M({p_1})} \over {M({p_2})}} = {{{p_1}{e^{{1 \over {{p_1}}}}}} \over {{p_2}{e^{{1 \over {{p_2}}}}}}}.
\end{split}
\end{equation}
Because $0<p_1<p_2$, we might as well take $x = {{{p_1}} \over {{p_2}}}$ and $x \in (0,1)$. Therefore, (\ref{equ:min}) can be written as
\begin{equation} \label{equ:root}
\begin{split}
{{M({p_1})} \over {M({p_2})}} = x{e^{{{1 - x} \over {{p_1}}}}} \ge x{e^{n\left( {1 - x} \right)}}.
\end{split}
\end{equation}
This means, we only need to check $x{e^{n\left( {1 - x} \right)}} \ge n-1$ holds if $n$ is large enough. It is easy to check that $x{e^{n\left( {1 - x} \right)}}$ monotonically increases in interval $(0,{1\over n})$ and monotonically decreases in the interval $({1 \over n},1)$. As a result, if  $\exists$ $ x_1 \in (0,{1\over n})$ and $x_2 \in ({1 \over n},1)$ s.t. $x_i{e^{n\left( {1 - x_i} \right)}}=n-1$ for $i=1,2$, then $x{e^{{{1 - x} \over {{p_1}}}}} \ge n-1$ in the interval $[x_1,x_2]$. 

Solve the equation $x_1{e^{n\left( {1 - x_1} \right)}} = n - 1$ in the interval $ (0,{1 \over n})$, we get
\begin{flalign}\label{equ:solve_eq2}
\ln x_1 + n(1 - x_1) &= \ln(n - 1)  \\
\ln x_1 + n -nx_1 &= \ln(n-1)  \tag{\theequation a}\label{equ:solve_eq2 a}\\
\ln x_1 + n&= \ln(n - 1) \tag{\theequation b}\label{equ:solve_eq2 b}\\
x_1&=e^{\ln(n-1)-n}. \tag{\theequation c}\label{equ:solve_eq2 c}
\end{flalign}
(\ref{equ:solve_eq2 b}) is obtained by removing $n x_1$ for the fact that $\mathop {\lim }\limits_{n \to \infty } {n e^{\ln (n - 1) - n}} = 0$, which leads to $n x_1$ is close to $0$ when $n$ is very large. 

Similarly, solve the equation $x_2{e^{n\left( {1 - x_2} \right)}} = n - 1$ in the interval $ ({1 \over n},1)$, we have
\begin{flalign}\label{equ:solve_eq1}
\ln x_2 + n(1 - x_2) &= \ln(n - 1)  \\
n - \ln(n - 1) & = {{n}}x_2 - \ln(1 + x_2 - 1) \tag{\theequation a}\label{equ:solve_eq1 a}\\
n - \ln(n - 1) & = {{n}}x_2 - \left( {(x_2 - 1) + o(x_2 - 1)} \right) \tag{\theequation b}\label{equ:solve_eq1 b}\\
n - \ln(n - 1) &= {{n}}x_2 - (x_2 - 1)  \tag{\theequation c}\label{equ:solve_eq1 c}\\
 x_2 &= 1 - {{\ln (n - 1)} \over {n - 1}}. \tag{\theequation d}\label{equ:solve_eq1 d}
\end{flalign}
(\ref{equ:solve_eq1 b}) follows from Taylor Series Expansion and (\ref{equ:solve_eq1 c}) is obtained by removing $o(x_2-1)$. In fact, it requires that $x_2-1$ is very close to $0$. Such a condition is satisfied because $\mathop {\lim }\limits_{n \to \infty }  { {{\ln (n - 1)} \over {n - 1}} }  = 0$ which leads to $x_2-1=- { {{\ln (n - 1)} \over {n - 1}} }\approx 0$ when $n$ is very large.

Based on the discussion above, we have $x{e^{n\left( {1 - x} \right)}} \ge n-1$ in the interval $[ {e^{\ln (n - 1) - n}},1-{{\ln (n-1)}\over{n-1}}]$. If $n$ tends to infinity, this interval will become as $(0,1)$. The proof is completed. 

\end{proof}
\begin{rem}
For arbitrary positive number $x<1$, one can always find a $N_0$, when $n>N_0$, $x \in [ {e^{\ln (n - 1) - n}},1-{{\ln (n-1)}\over{n-1}}]$ holds.
\end{rem}

\section{Proof of Lemma \ref{lem:root1}}\label{Appendices B}
\begin{proof}
We prove this lemma in a similar way as Appendix \ref{Appendices A}. We also take $x = {{{p_1}} \over {{p_2}}}$ and $x \in (0,1)$. According to (\ref{equ:root}),
it is easy to check that $ x{e^{{{1 - x} \over {{p_1}}}}}$ monotonically increases in interval $(0,p_1)$ and monotonically decreases in the interval $(p_1,1)$. As a result, if  $\exists$ $ x_1 \in (0,p_1)$ and $x_2 \in (p_1,1)$ s.t. $ x_i{e^{{{1 - x_i} \over {{p_1}}}}}=n-1$ for $i=1,2$, then $x{e^{{{1 - x} \over {{p_1}}}}} \ge n-1$ in the interval $[x_1,x_2]$.

In a similar way as Appendix \ref{Appendices A}, we have
\begin{flalign}\label{equ:B solve_equation}
 x_1&=e^{\ln(n-1)-{1\over p_1}},   \\
 x_2 &= 1 - \frac{p_1}{1-p_1} \ln (n-1).  \tag{\theequation a}\label{equ:B solve_equation a}
\end{flalign}
(\ref{equ:B solve_equation}) is obtained by removing $x_1/p_1$ for the fact that $\mathop {\lim }\limits_{{p_1} \to 0} {1 \over {{p_1}}}{e^{\ln (n - 1) - {1 \over {{p_1}}}}} = 0$ which leads to $x_1/p_1$ is close to $0$ when ${p_1} \ll {1 \over {1 + \ln (n - 1)}}$. (\ref{equ:B solve_equation a}) requires $x_2-1$ is close to $0$. Such a condition is satisfied because $\frac{p_1}{1-p_1} \ln (n-1)\to0$ when ${p_1} \ll {1 \over {1 + \ln (n - 1)}}$, which leads to $x_2-1=- \frac{p_1}{1-p_1} \ln (n-1)\approx 0$.

Based on the discussion above, we have $ x{e^{{{1 - x} \over {{p_1}}}}}\ge n-1$ in the interval $[e^{\ln(n-1)-{1\over p_1}},1 - \frac{p_1}{1-p_1} \ln (n-1)]$. When ${p_1} \ll {1 \over {1 + \ln (n - 1)}}$, this interval will become as $(0,1)$. 
\end{proof}
\begin{rem}
For arbitrary positive number $x<1$ and constant $n$, one can always find a $p_0$, when $p_1<p_0$, $x \in[e^{\ln(n-1)-{1\over p_1}},1 - \frac{p_1}{1-p_1} \ln (n-1)]$ holds.
\end{rem}

\section{Proof of Proposition \ref{thm:code1}}\label{Appendices C}
\begin{proof}
We can write the constrained minimization $\mathcal{P}_1$ using Lagrange multipliers as the minimization of
\begin{equation}
\varphi  = \sum\limits_{i=1}^n {{p_i}{{e^{{1-p_i} \over {p_i}}}  {{l_i}^{-1}}}}  + \lambda \left( {\sum\limits_{i=1}^n {{l_i}}  - K} \right).
\end{equation}
Differentiating with respect to $l_i$, we get
\begin{equation}
{{\partial \varphi } \over {\partial {l_i}}} =  - {{M({p_i})} \over {l_i^2}} + \lambda
\end{equation}
where $M(p_i)=p_i e^{{1-p_i} \over {p_i}}$ according to (\ref{equdefn:MIM of Event}).

Setting the derivative to $0$, and we get ${l_i} =  {\sqrt {{{M({p_i})} / \lambda }} } $. Substituting this in the constraint $K = \sum\limits_{i=1}^n {{l_i}}$,  we get $\lambda  = {\left( {{{\sum\limits_{i = 1}^n {\sqrt {M({p_i})} } } \mathord{\left/
 {\vphantom {{\sum\limits_{i = 1}^n {\sqrt {M({p_i})} } } K}} \right.
 \kern-\nulldelimiterspace} K}} \right)^2}$, and hence
 \begin{equation}
{l_i} = {{{\sqrt {M({p_i})} } / ({\sum\limits_{i=1}^n { {\sqrt {M({p_i})} }  } }}K}).
\end{equation}
But since $l_i$ must be intergers, we will set $ {l_i} = \left\lceil{{{\sqrt {M({p_i})} } / ({\sum\limits_{i=1}^n { {\sqrt {M({p_i})} }  } }}K}) \right\rceil$.
\end{proof}

\section{Proof of Proposition \ref{thm:code2}}\label{Appendices D}
\begin{proof}
We can write the constrained minimization $\mathcal{P}_2$ using Lagrange multipliers as the minimization of
\begin{equation}
\varphi  =  {\sum\limits_{i = 1}^n {{p_i}} {e^{{{1 - {p_i}} \over {{p_i}}}}}{\gamma ^{ - {l_i}}}} + \lambda \left( {\sum\limits_{i=1}^n {{l_i}}  - K} \right).
\end{equation}
Differentiating with respect to $l_i$, we get
\begin{equation}
{{\partial \varphi } \over {\partial {l_i}}} =  - M\left( {{p_i}} \right){\gamma ^{ - {l_i}}}\ln \gamma  - \lambda  
\end{equation}
where $M(p_i)=p_i e^{{1-p_i} \over {p_i}}$ according to (\ref{equdefn:MIM of Event}). 

Setting the derivative to $0$, and we get ${l_i} ={{{\ln M\left( {{p_i}} \right) + \ln \ln \gamma  - \ln \left( { - \lambda } \right)} \over {\ln \gamma }} }$. 
In order to ensure $l_i>0$, so we take
\begin{equation}
 {l_i} = {\left( {{\ln M\left( {{p_i}} \right) + \ln \ln \gamma  - \ln \left( { - \lambda } \right)} \over {\ln \gamma }}\right) }^+.
 \end{equation}
Substituting this in the constraint $K = \sum\limits_{i=1}^n {{l_i}}$,  we find $\ln \left( { - \lambda } \right) = {{\left( {\tilde N\ln\ln\gamma  - K\ln \gamma  + \sum\limits_{i = 1}^{\tilde N} {\ln M({{\tilde p}_i})} } \right)} \mathord{\left/
 {\vphantom {{\left( {\tilde Nlnln\gamma  - K\ln \gamma  + \sum\limits_{i = 1}^{\tilde N} {\ln M({{\tilde p}_i})} } \right)} {\tilde N}}} \right.
 \kern-\nulldelimiterspace} {\tilde N}}$,
where ${\tilde N}$ is the number of $l_i$ which is greater than zero and $\{\tilde p_i, i=1,2,...,\tilde N\}$ is part of the commutative sequence of $\{ p_i, i=1,2,...,N\}$ in increment order which satisfies ${{\ln M\left( {{p_i}} \right) + \ln \ln \gamma  - \ln \left( { - \lambda } \right)} }>0$. Since $l_i$ must be intergers, the optimal code lengths are 
 \begin{equation}
{l_i} = {\left\lceil {{{\left( {\ln M({p_i}) - \sum\limits_{i = 1}^{\tilde N} {\ln M({{\tilde p}_i})} /\tilde N + K\ln \lambda /\tilde N} \right)} \mathord{\left/
 {\vphantom {{\left( {\ln M({p_i}) - \sum\limits_{i = 1}^{\tilde N} {\ln M({{\tilde p}_i})} /\tilde N + K\ln \lambda /\tilde N} \right)} {\ln \lambda }}} \right.
 \kern-\nulldelimiterspace} {\ln \gamma }}} \right\rceil ^ + }
\end{equation}
\end{proof}

\section{Proof of Proposition \ref{thm: BSC}}\label{Appendices E}
\begin{proof}
It is obtained that
\begin{flalign}\label{equ: E LP}
  L(p(1 - \varepsilon ) + \varepsilon(1 - p)) - L(p)   &= L(p + \varepsilon (1 - 2p)) - L(p)  \\
  &  = L(p) - {{1 - p} \over {{p^2}}}\varepsilon (1 - 2p) + o\left( {\varepsilon (1 - 2p)} \right) - L(p) \tag{\theequation a} \label{equ: E LP a} \\
  &  =  - {{1 - p} \over {{p^2}}}\varepsilon (1 - 2p) + o\left( {\varepsilon (1 - 2p)} \right) \tag{\theequation b} \label{equ: E LP b} \\
  &   \approx  - {{1 - p} \over {{p^2}}}\varepsilon (1 - 2p) . \tag{\theequation c} \label{equ: E LP c}
\end{flalign}
\eqref{equ: E LP a} follows from Taylor series expansion in (\ref{equ:taylor series expansion}) and \eqref{equ: E LP c} is obtained by removing $o\left( {\varepsilon (1 - 2p)} \right)$. In fact,
it requires that $o\left( {\varepsilon (1 - 2p)} \right)$ is very close to 0, which is satisfied because $\varepsilon \ll p$ and $p\ne 1/2$. Therefore, we find 
\begin{equation}\label{equ:C0}
\log M(p(1 - \varepsilon ) + \varepsilon(1 - p)) - \log M(p) =  - {{1 - p} \over {{p^2}}}\varepsilon (1 - 2p).
\end{equation}
Hence, we obatin
\begin{equation}\label{equ:C1}
M(p(1 - \varepsilon ) + \varepsilon(1 - p)) = {e^{ - {{1 - p} \over {{p^2}}}\varepsilon (1 - 2p)}}M(p).
\end{equation}
Likewise, we obtain
\begin{equation}\label{equ:C2}
M(p\varepsilon  + (1 - p)(1 - \varepsilon )) = {e^{{p \over {{{\left( {1 - p} \right)}^2}}}\varepsilon (1 - 2p)}}M(1 - p{\rm{)}}.
\end{equation}
Put (\ref{equ:C1}) and (\ref{equ:C2}) into $\psi(X,Y)$ and get
{\small
\begin{flalign}\label{equ:BSC_proof}
\psi(X,Y) & = \log\left( {M(p) + M(1 - p)} \right) - \log \left( {{e^{ - {{1 - p} \over {{p^2}}}\varepsilon (1 - 2p)}}M(p) + {e^{{p \over {{{\left( {1 - p} \right)}^2}}}\varepsilon (1 - 2p)}}M(1 - p{\rm{)}}} \right)  \\
      &  = \log\left( {M(p) + M(1 - p)} \right) - \log \left( {{e^{ - {{1 - p} \over {{p^2}}}\varepsilon (1 - 2p)}}\left( {M(p) + M(1 - p{\rm{) + }}({e^{{{3{p^2} - 3p + 1} \over {{p^2}{{\left( {1 - p} \right)}^2}}}\varepsilon (1 - 2p)}} - 1)M(1 - p{\rm{)}}} \right)} \right)  \tag{\theequation a}\label{equ:BSC_proof a}\\
        &  \approx \log\left( {M(p) + M(1 - p)} \right) - \log \left( {{e^{ - {{1 - p} \over {{p^2}}}\varepsilon (1 - 2p)}}\left( {M(p) + M(1 - p{\rm{)}}} \right)} \right)  \tag{\theequation b}\label{equ:BSC_proof b}\\
  &  = \log\left( {M(p) + M(1 - p)} \right) + {{1 - p} \over {{p^2}}}\varepsilon (1 - 2p) - \log \left( {M(p) + M(1 - p)} \right)  \tag{\theequation c}\label{equ:BSC_proof c}\\
  &  = {{(1 - p )(1 - 2p) } \over {{p^2}}}\varepsilon. \tag{\theequation d}\label{equ:BSC_proof d}
\end{flalign}
}
 \eqref{equ:BSC_proof b} is obtained by removing ${({e^{{{3{p^2} - 3p + 1} \over {{p^2}{{\left( {1 - p} \right)}^2}}}\varepsilon (1 - 2p)}} - 1)M(1 - p{\rm{)}}}$ because it is mush smaller than $\left( {M(p) + M(1 - p)} \right)$. Generally speaking, such a condition is satisfied in this scenario, which is shown in Fig. \ref{fig:channelBSC}.

For convenience, we define $s(p)={{(1 - p )(1 - 2p) } \over {{p^2}}}$, and we obtain 
\begin{equation}
\psi(X,Y)   ={ \mathscr{L}({\textbf{p}_x}) -  \mathscr{L} ({\textbf{p}_{y}})}\approx  \varepsilon s(p).
\end{equation}
The proof is completed. 

For convenience, we also define approximate value of MIM change as
{\small
\begin{equation}\label{equ:approx MIM change}
\psi_{approx}(X,Y) = \log\left( {M(p) + M(1 - p)} \right) - \log \left( {{e^{ - {{1 - p} \over {{p^2}}}\varepsilon (1 - 2p)}}\left( {M(p) + M(1 - p{\rm{) + }}({e^{{{3{p^2} - 3p + 1} \over {{p^2}{{\left( {1 - p} \right)}^2}}}\varepsilon (1 - 2p)}} - 1)M(1 - p{\rm{)}}} \right)} \right).
\end{equation}
}
\end{proof}

\section{Proof of Proposition \ref{thm:Rate MIM LOSS}}\label{Appendices F}
\begin{proof}
\begin{flalign}\label{equ:Rate MIM LOSS Fun}
{R^{(MIM)}}(D) &= \mathop {\max }\limits_{p\left( {\hat x|x} \right):\sum\nolimits_{(x,\hat x)} {p(x)p(\hat x|x)d(x,\hat x) \le D} } \left\{ { \mathscr{L}({\textbf{p}_x}) -  \mathscr{L} ({\textbf{p}_{\hat x}})} \right\} \\
& =  \mathscr{L}({\textbf{p}_x})  - \mathop {\min }\limits_{p\left( {\hat x|x} \right):\sum\nolimits_{(x,\hat x)} {p(x)p(\hat x|x)d(x,\hat x) \le D} } \left\{ {\mathscr{L}({\textbf{p}_{\hat x}})} \right\}.\tag{\theequation a} \label{equ:Rate MIM LOSS Fun a}
\end{flalign}
The key of this optimization problem is the second item in \eqref{equ:Rate MIM LOSS Fun a}. Define $p_{ij}=p(i,j)$ for convenience, and we can solve it as a constrained minimization
\begin{flalign}\label{equ:Rate MIM LOSS optization}
\mathcal{P}_{4-A}: \,\,\mathop {\min }\limits_{{p_{ij}}}\,\,\, &({p_{00}} + {p_{10}}){e^{{1 \over {{p_{00}} + {p_{10}}}}}} + ({p_{01}} + {p_{11}}){e^{{1 \over {{p_{01}} + {p_{11}}}}}}  \\
\textrm{s.t.}\,\,\,& {{p_{00}}{\rm{ + }}{{\rm{p}}_{01}} = p}, \quad {{p_{10}}{\rm{ + }}{{\rm{p}}_{11}} = 1 - p},\quad {{{\rm{p}}_{01}} + {p_{10}} \le D} \tag{\theequation a}\label{equ: Rate MIM LOSS optization a}\\
&  p_{ij}>0,\,\,\, i=0,1 \,\,j=0,1. \tag{\theequation b}\label{equ: Rate MIM LOSS optization b}
\end{flalign}
Simplify it, we have
\begin{flalign}\label{equ:Rate MIM LOSS optization Simplify}
\mathcal{P}_{4-B}: \,\,\mathop {\min }\limits_{{p_{01}},{p_{10}}}\,\,\, &(p - {p_{01}} + {p_{10}}){e^{{1 \over {p - {p_{01}} + {p_{10}}}}}} + ({p_{01}} + 1 - p - {p_{10}}){e^{{1 \over {{p_{01}} + 1 - p - {p_{10}}}}}} \\
\textrm{s.t.}\,\,\,& {{{{p}}_{01}} + {p_{10}} \le D},\,\,\,p_{01}>0,\,\,\, p_{10}>0  .\tag{\theequation a} \label{equ: Rate MIM LOSS optization Simplify a}
\end{flalign}
Without loss of generality, letting $t=p - {p_{01}} + {p_{10}}$, this minimization problem $\mathcal{P}_{4-B}$ can be rewritten as 
\begin{equation}
\mathop {\min }\limits_t \left\{ {t{e^{{1 \over t}}} + (1 - t){e^{{1 \over {1 - t}}}}} \right\}.
\end{equation}
The function $f(t)={t{e^{{1 \over t}}} + (1 - t){e^{{1 \over {1 - t}}}}}$ is symmetrical about $t={1\over2}$. It achieves the minimum when $t={1\over2}$ and it monotonically decreases in $(0,{1\over2})$. By linear programming, when ${{{{p}}_{01}} + {p_{10}} \le D}$, the range of $t$ is
\begin{equation}
t \in  \left\{
   \begin{aligned}
 & (p-D,p+D),0\le D< \min\{ p,1-p\} ,\\
 & (0,p+D),p\le D \le 1-p,\\
 & (p-D,1),1-p\le D \le p,\\
 & (0,1), D > \max\{ p,1-p\} .\\
   \end{aligned}
   \right. 
\end{equation}
Based on the discussion above, the NMIM loss distortion function for a binary source is
\begin{flalign}\label{equ:rate dis}
{R^{(MIM)}}(D)=
  \left\{
   \begin{aligned}
 & \log({p{e^{{1 \over p}}} + (1 - p){e^{{1 \over {1 - p}}}}})-\log((p + D){e^{{1 \over {p + D}}}} + (1 - p - D){e^{{1 \over {1 - p - D}}}}),\\
 &p+D<{1\over2}, 0\le D\le \min\{ p,1-p\} \quad or \quad p \le D \le 1-p,\\
 & \log({p{e^{{1 \over p}}} + (1 - p){e^{{1 \over {1 - p}}}}})-\log((p - D){e^{{1 \over {p - D}}}} + (1 - p + D){e^{{1 \over {1 - p + D}}}}),\\
 &p-D>{1\over2}, 0\le D\le \min\{ p,1-p\} \quad or \quad 1-p \le D \le p,\\
 & \log({p{e^{{1 \over p}}} + (1 - p){e^{{1 \over {1 - p}}}}})-\log{e^2}, else.
   \end{aligned}
   \right. 
\end{flalign}
For convenience, We define
\begin{equation}\label{equ:delta_p0}
\delta(p)= \log({p{e^{{1 \over p}}} + (1 - p){e^{{1 \over {1 - p}}}}})-\log{e^2}.
\end{equation}
Because $D\ge0$, $p+D<{1\over2}$ leads to $p<{1\over2}$. Similarly, $p-D>{1\over2}$ is equivalent to $p>{1\over2}$. Hence, (\ref{equ:rate dis}) can be rewritten as
\begin{flalign}
{R^{(MIM)}}(D)=
  \left\{
   \begin{aligned}
 &\log{e^2}+\delta(p)-\log((p + D){e^{{1 \over {p + D}}}} + (1 - p - D){e^{{1 \over {1 - p - D}}}}), p+D<{1\over2}, 0\le D \le 1-p,\\
 &\log{e^2}+\delta(p)-\log((p - D){e^{{1 \over {p - D}}}} + (1 - p + D){e^{{1 \over {1 - p + D}}}}), p-D>{1\over2}, 0\le D \le p,\\
 &\delta(p), else.
   \end{aligned}
   \right.
\end{flalign}
\end{proof}

\section{Proof of Proposition \ref{thm:HC1}}\label{Appendices G}
\begin{proof}
According to \cite{Elements}, $H(X)\triangleq H(p)$ is a monotone increasing function in $(0,0.5]$. Hence, $H(X)\le Ct$ leads to $0<p\le p_0 \le 0.5$ where $H(p_0)=Ct$.

$H(Y)=H(p+D)$. Letting $q=p+D$, we can write the constrained maximization $\mathcal{P}_3$ as the maximization of $H(q)$. Obviously, $H(q)$ increases in $(0,0.5]$ and decreases in $(0.5,1)$.

Actually, $R^{(MIM)}$ is the maximum of $\psi(X,Y)$. Refer to \ref{sec:MIM Loss} and substitute this in the constraint $\psi(X,Y) \le \Delta$, we obtain
\begin{equation}
\left\{\begin{array}{l@{\hspace{1cm}}l}
 R^{(MIM)}(D) =0,  \quad if \quad \Delta =0\\
R^{(MIM)}(D)\le \Delta, \quad if \quad 0<\Delta < \delta(p_0)\\
 R^{(MIM)}(D) =\delta(p_0), \quad else\\
\end{array}\right.\Rightarrow \left\{\begin{array}{l@{\hspace{1cm}}l}
D=0, \quad if \Delta=0\\
0\le D < D^{(MIM)}(\Delta),\,if \quad 0<\Delta < \delta(p_0)\\
D\ge 0, \quad \Delta \ge \delta(p_0).\\
\end{array}\right.
\end{equation}
where $D^{(MIM)}$ is the inverse function of $R^{(MIM)}(D)$.
Substituting it in $q=p_0+D$, we obtain
\begin{equation}\label{equ:E2}
\begin{split}
 \left\{
   \begin{aligned}
&q=p_0,  \quad if \quad \Delta =0\\
& p_0<q<p_0+D^{(MIM)}(\Delta)\le0.5,  \quad if \quad  0<\Delta < \delta(p_0)\\
& q>p_0,  \quad if \quad \Delta \ge \delta(p_0) \\
   \end{aligned}
   \right.,
   \end{split}
\end{equation}
where $p_0+D^{(MIM)}(\Delta)\le0.5$ is due to
\begin{equation}
p_0+D^{(MIM)}(\Delta)\le p_0+D^{(MIM)}(\delta(p_0))=p_0+0.5-p_0=0.5.
\end{equation}
Hence, the maximum value of $H(Y)$ is
\begin{equation}\label{equ:E3}
\begin{split}
\mathop {\max }\limits_Y H(Y) =\left\{
   \begin{aligned}
   &H(p_0)=Ct,  \quad if \quad \Delta =0,\\
&H(p_0+D^{(MIM)}(\Delta)),  \quad if \quad  0<\Delta < \delta(p_0),\\
&H(0.5)=1 ,  \quad if \quad \Delta \ge \delta(p_0). \\
   \end{aligned}
   \right.
   \end{split}
\end{equation}

\end{proof}






\end{document}